\documentclass[12pt]{article}
\usepackage[colorlinks,pagebackref=false]{hyperref}
\hypersetup{citecolor=Blue, linkcolor=blue, urlcolor = blue}
\usepackage{url}

\usepackage{amsfonts,amsmath,amsthm,amssymb,verbatim,bm}
\usepackage{graphicx}
\usepackage{xcolor}
\usepackage{natbib}
\defcitealias{bernanke2009}{Ben Bernanke}

\usepackage{setspace}
\usepackage{epsf}
\usepackage{fullpage}
\usepackage{bbm}
\usepackage{enumerate}
\usepackage{paralist}
\usepackage[margin=10pt,labelfont=bf,labelsep=endash]{caption}
\usepackage{subcaption}
\usepackage{floatrow}
\usepackage{multirow}
\usepackage{footmisc}
\usepackage{xr}
\usepackage{tikz}
\usetikzlibrary{arrows,decorations,decorations.pathreplacing,calc,matrix}
\definecolor{Red}{rgb}{1,0,0}
\definecolor{Blue}{rgb}{0,0,1}
\definecolor{Green}{rgb}{0,1,0}
\definecolor{magenta}{rgb}{1,0,.6}
\definecolor{lightblue}{rgb}{0,.5,1}
\definecolor{lightpurple}{rgb}{.6,.4,1}
\definecolor{gold}{rgb}{.6,.5,0}
\definecolor{orange}{rgb}{1,0.4,0}
\definecolor{hotpink}{rgb}{1,0,0.5}
\definecolor{newcolor2}{rgb}{.5,.3,.5}
\definecolor{newcolor}{rgb}{0,.3,1}
\definecolor{newcolor3}{rgb}{1,0,.35}
\definecolor{darkgreen1}{rgb}{0, .35, 0}
\definecolor{darkgreen}{rgb}{0, .6, 0}
\definecolor{darkred}{rgb}{.75,0,0}
\definecolor{aurometalsaurus}{rgb}{0.43, 0.5, 0.5}

\newcounter{desccount}

\newcommand{\dref}[1]{\hyperref[#1]{#1}}

\newcommand{\set}[1]{\{#1\}}

\newcommand{\bea}{\begin{align*}}
\newcommand{\eea}{\end{align*}}

\renewcommand{\bar}{\overline}
\renewcommand{\paragraph}[1]{\medskip \noindent \textit{#1}}

\def\e{\varepsilon}

\def\th{\theta}
\def\t{\theta}

\def\l{\lambda}
\def\T{\Theta}
\def\E{\mathbb{E}}

\def\bq{\begin{equation}}
\def\eq{\end{equation}}
\def\ba{\begin{eqnarray}}
\def\ea{\end{eqnarray}}

\definecolor{clemson-orange}{RGB}{234,106,32}
\definecolor{chicago-maroon}{RGB}{128,0,0}
\definecolor{northwestern-purple}{RGB}{82,0,99}
\definecolor{cornell-red}{RGB}{179,27,27}
\definecolor{sauder-green}{RGB}{171,180,0}
\definecolor{gray}{RGB}{192,192,192}
\definecolor{lawngreen}{RGB}{0,250,154}
\definecolor{pink}{RGB}{255,0,128}

\newtheorem{theorem}{Theorem}
\newtheorem{claim}{Claim}

\newtheorem{corollary}{Corollary}
\newtheorem{example}{Example}
\newtheorem{lemma}{Lemma}

\newtheorem{remark}{Remark}

\renewcommand\thesection{\Roman{section}}

\begin{document}

\title{\bf Bailout Stigma}
\author{\textsc{Yeon-Koo Che} \and \textsc{Chongwoo Choe} \and \textsc{Keeyoung Rhee}\thanks{Che: Department
of Economics, Columbia University (email:
\href{mailto:yeonkooche@gmail.com}{yeonkooche@gmail.com}); Choe: Centre for Global Business and 
Department of Economics, Monash University (email:
\href{mailto:chongwoo.choe@monash.edu}{chongwoo.choe@monash.edu});
Rhee: Department of Economics, Sungkyunkwan University (SKKU) (email: \href{mailto:ky.rhee829@gmail.com}{ky.rhee829@gmail.com}).
 The authors are grateful to anonymous referees, an Associate Editor, Olivier Armantier, Philip Bond (the Editor), Xuandong Chen, Keshav Dogra, Oliver Giesecke, Spencer Kwon, Haaris Mateen, Aniko \"Ory, Joel Shapiro, Jean Tirole, Haoxiang Zhu as well as participants at numerous seminars and conferences, for their comments.  
}  
}
\date{\today}

\maketitle
\begin{abstract}  
We develop a model of bailout stigma where accepting a bailout signals a firm’s balance-sheet weakness and worsens its funding prospect. To avoid stigma, high-quality firms either withdraw from subsequent financing after receiving bailouts or refuse bailouts altogether to send a favorable signal. The former leads to a short-lived stimulation with a subsequent market freeze even worse than if there were no bailouts. The latter revives the funding market, albeit with delay, to the level achievable without any stigma, and implements a constrained optimal outcome. A menu of multiple bailout programs also compounds bailout stigma and worsens market freeze.

	\medskip
	
	

\end{abstract}

\noindent History is fraught with financial crises and large-scale government interventions, the latter often involving a highly visible and significant wealth transfer from taxpayers to banks and their creditors. As a recent example, during the 2007-2009 Great Recession, the US government paid \$125 billion for assets worth \$86-109 billion to the nine largest banks under the Troubled Asset Relief Program (TARP) (\citet{veronesi2010paulson}).
	A rationale for such public interventions -- to be called bailouts throughout the paper -- is that the government can jump-start a market that would otherwise freeze due to adverse selection. By cleaning up bad assets, or ``dregs skimming,'' through public bailouts, the government can improve market confidence, thereby galvanizing transactions in healthier assets (\citet{philippon2012optimal}, \citet{tirole2012overcoming}). However, the flip side of such dregs-skimming is that bailouts can attach stigma to their recipients, increasing future borrowing costs. The fear of this stigma may in turn discourage financially distressed firms from accepting bailout offers in the first place, undermining the effectiveness of such interventions.\footnote{Such a concern is echoed in a speech by the former Federal Reserve chairman \citetalias{bernanke2009} in 2009: ``The banks' concern was that their recourse to the discount window, if it became known, might lead market participants to infer weakness---the so-called stigma problem.'' Several anecdotes suggest that this concern is well-founded. Ford refused rescue loans under the Auto Industry Program in the TARP, with a view to ``legitimately portraying itself as the healthiest of Detroit's automakers'' (``A risk for Ford in shunning bailout, and possibly a reward,'' {\it The New York Times}, December 19, 2008). It is also well known that Jamie Dimon, CEO of JP Morgan Chase, wanted to exit the TARP to avoid the stigma (``Dimon says he's eager to repay `Scarlet Letter' TARP,'' {\it Bloomberg}, April 16, 2009). The fear of stigma is not the only reason for an early exit. \citet{wilson2012escaping} find that early exit by banks is also related to CEO pay, bank size, capital, and other financial conditions. These factors, especially the restrictions on executive compensation, are also found to cause the reluctance to accept   bailouts (\cite{bayazitova}, \cite{cadman}). The literature documenting the empirical evidence on stigma is reviewed in Section \ref{sec:lit}.}$^{,}$\footnote{Policy makers during the Great Recession were aware of such a fear and took efforts to alleviate the stigma.  At the now-famous meeting held on October 13, 2008, Henry Paulson, then Secretary of the Treasury,  ``compelled'' the CEOs of the nine largest banks to be the initial participants in the TARP, precisely to eliminate the stigma (``Eight days: the battle to save the American financial system,'' {\it The New Yorker}, September 21, 2009). The rates at the Fed's discount window, usually set above the federal funds rate, were cut half a percentage point to counteract the stigma attached to using the window (\citet[p.~129]{geithner2015}).}

 The purpose of this paper is to study a dynamic model of bailouts to understand how reputational concerns affect the effectiveness of  public bailouts and how an optimal  bailout policy should be designed in light of these concerns. {To address these questions, we present a two-period model of a bailout program in which the government offers to purchase assets from firms. The focus on asset-purchase as the bailout tool allows us to capture firms' reputational concerns in a simple parsimonious way.} There is a continuum of firms, each with one unit of a legacy asset in each period. {The quality of asset in both periods is identical for each firm and is its private information.} 
In each period, firms have access to profitable investment opportunities. However, liquidity constraints and the non-pledgeability of projects imply that firms need to sell assets 
 to fund those projects. As in \citet{tirole2012overcoming}, we focus on the case where adverse selection leads to a market freeze and calls for a public bailout.  To understand the reputational consequence of accepting a bailout, we assume that the government runs a bailout program in the first period only. 
 The sale of assets in the market is not publicly observed, but firms' acceptance of the bailout offer is observed. The market updates its belief on the quality of assets based on the observation of firms' decisions 
 in the first period, and thus makes its second-period offer accordingly.


Bailout stigma is captured in our model  by the unfavorable terms  that bailout  recipients  suffer from the sale of their assets in the second period. As shown by \citet{philippon2012optimal} and \citet{tirole2012overcoming}, a key function of a public bailout  is dregs-skimming: by taking out the left tail of the worst quality assets, a bailout improves the perceived quality of remaining assets, thereby rejuvenating asset trade in the market. But this means that the market will believe those that accept the bailout have worse assets than those that refuse the bailout. Such a belief is reflected in the differential treatment of the two groups of firms in the second period: the market's offer to bailout recipients would be in general in worse terms than that made to the bailout holdouts. 

	The precise mechanism by which the bailout stigma affects the efficacy of the bailout crucially depends on firms' strategic responses to the bailout offer, manifested in two possible equilibrium outcomes---{\bf short-lived stimulation} and {\bf delayed stimulation}.  
	
A short-lived stimulation equilibrium arises when high-quality firms strategically avoid the bailout stigma by accepting a bailout in the first period but withdrawing from the market in the second period.  Suppose that high-type firms indeed  behave in this way. Then, the market in $t=2$ would believe that the bailout recipients that \emph{do} participate in the second-period have bad assets, and assign stigma to them. To avoid that stigma,  firms would rather choose to sell their assets to the market at discount in the first period. This in turn deteriorates their sale terms, ultimately inflicting a commensurate haircut on the market-sellers; in effect, the bailout stigma ``spreads'' from bailout recipients to market sellers in this equilibrium.  This contagion of stigma in turn leads firms with high-quality assets 
to accept a bailout offer  but withdraw from the second-period market 
altogether to avoid that stigma, thus validating the initial hypothesis.\footnote{The withdrawal from the second-period funding market should not be  literally  interpreted as a firm exiting from the funding market altogether.  Our model is (inevitably) stylized, so  the results should be interpreted with some care.  In practice, a firm's withdrawal from the funding market will more realistically correspond to its cutting back on additional projects \emph{at the margin} that it would otherwise pursue.}  

 {The presence of high-type firms that 
 withdraw from the subsequent market/financing 
	undermines the overall efficacy of the bailout.  
	Their withdrawal exacerbates the stigma for those participating in the subsequent market, which has the chilling effect on the asset markets in both periods.  Unlike \cite{tirole2012overcoming} where the government intervention boosts the confidence in private markets, intervention now undermines market confidence. The consequence is devastating: \emph{the market freeze is even worse than if there were no bailout!} However, this does not mean that a bailout has no effect. The increased investment by bailout recipients in the first period may outweigh the dampening effect on the markets.}  Nevertheless, stimulation is short-lived in this equilibrium. We show that the policy maker can avoid this equilibrium by offering sufficiently generous bailout terms---that is, high purchase prices for assets, in which case the stigma will manifest itself in a different form: ``delayed stimulation.''
	
A delayed stimulation equilibrium arises when high-quality firms refuse to sell either to the government or to the market in the first period.  They do so to build a good reputation in order to sell their assets in the second period at favorable terms. This equilibrium is possible when the bailout offer is generous enough to attract a large fraction of firms with low-quality assets. This allows firms to send a credible signal about their asset quality when they \emph{refuse} the bailout offer. Buyers will then respond with  a very attractive price offer in the second period---one that makes it worthwhile for high-quality firms to forego asset sale in the first period.  In sum, the equilibrium endogenously creates an opportunity for high-quality firms to signal their financial strength by rejecting the government's generous offer.

	Such a favorable signaling opportunity offsets the adverse effect of   bailout stigma, even though market rejuvenation is delayed to the second period. The presence of firms rejecting a bailout means that the volume of asset trade in the first period is lower than that of the one-period benchmark with the same bailout offer. In an extreme case, it is even possible that the bailout has no stimulation effect in the first period relative to the {\it laissez-faire} economy. Such an initial lack of response may be seen as a policy failure.  However, the policy ``quietly'' strengthens market confidence in the asset quality of refusing firms, which bears dividends in the second period.  In fact, the overall trade volume is higher than that in a short-lived stimulation equilibrium; remarkably, it is the same \emph{as if there were no bailout stigma}---that is, if the identities of bailout recipients were concealed successfully, which is often difficult to achieve in practice.  {Strikingly, a bailout offer is more effective when many firms reject it (to build a favorable reputation) rather than accept it. Rejection of bailouts   could therefore be a blessing in disguise.}

	
 {The desirability of the delayed stimulation equilibrium is further reinforced when the costs of bailouts are taken into account. When such social costs are included in the model, we show that a delayed stimulation equilibrium (as well as a secret bailout) induced by a suitable bailout offer is constrained optimal among all bailout mechanisms offered in $t=1$, including those involving a menu of multiple bailout terms.}  By contrast, a short-lived stimulation equilibrium is strictly suboptimal, regardless of the bailout terms.  This is because the severe stigma arising from that equilibrium requires a high deficit, and a correspondingly high social cost, to generate the same degree of stimulation in trade. To eliminate the possibility of this less desirable equilibrium, the policy maker may wish to make the terms of bailout even more generous than would be optimal. 
	This approach, although departing from the classical Bagehot's rule,\footnote{Bagehot's rule, orginating from the 1873 book, \emph{Lombard Street}, by William Bagehot, prescribes that central banks should charge a higher rate than the markets to discourage banks from borrowing once the crisis subsides. Bailout stigma was not a serious issue in 1873, however, since the regulatory system in 1873 Britain ensured concealment of the identities of emergency borrowers, as \citet{gorton2015stress} points out.} is consistent with the approach taken by policy makers during the Great Recession.   Finally, the optimality of a single bailout program calls into question the wisdom of offering a menu of multiple bailouts, which can never be socially desirable.  Multiple programs give firms increased opportunities and the incentives for signaling, which compounds the overall level of stigma and market freeze.

The remainder of the paper is organized as follows.  Section \ref{sec:model} presents our model while Section \ref{sec:benchmark} analyses several benchmark cases.  In Section \ref{sec:gov}, we study various equilibria under government intervention. Section \ref{sec:wel} provides the welfare analysis of bailout policy.  Section \ref{sec:lit} discusses related literature.   Section \ref{sec:conclusion} concludes.  Proofs not contained in the main text are deferred to Appendix and Online Appendix.

\section{Model} \label{sec:model}

We adopt a model that extends \citet{tirole2012overcoming} to a setup which admits a bailout stigma. There is a continuum of firms each endowed with two units of legacy assets with the same value; one unit of the asset becomes available in each of two periods ($t = 1, 2$) for possible sale.\footnote{One can think of the assets as account receivable or the contract for (securitized) assets to be delivered over two periods.}   The asset value $\t$ of each firm  is privately known to that firm and distributed on $[0,1]$ according to cdf $F$ with density $f$.   For convenience, we hereafter call a firm with legacy asset $\t$ a type-$\t$ firm. Throughout, we assume that $f$ is log-concave, that is, $d^2 \log f(\t)/d \t^2<0$. Log-concavity of $f$ implies intuitive properties we will use on  truncated conditional expectation: for any $0<a<b<1$, $0 < \frac{\partial}{\partial a} \mathbb{E}[\th | a \leq \th \leq b],\frac{ \partial }{\partial b} \mathbb{E}[\th | a \leq \th \leq b] < 1$ (see \cite{bagnoli2005log}).  We additionally assume that for each $b\in (0,1]$,   $2\E[\t | a < \t < b] - \E[ \t | \t \le a]$ is increasing in $a$ for any  $a\in (0,b)$. These properties, which hold for many well-known distributions, facilitate the characterization of our equilibria.

In each period, an investment project  becomes available to each firm.   The project is socially valuable with  {\it net} return $S>0$ but requires  funding of $I>0$. The firm can finance the project by selling its legacy asset each period.  As we will see, the outcome from this {\it laissez-faire} regime will typically be inefficient due to the adverse selection associated with uncertain asset value.  This inefficiency rationalizes a government bailout in the form of an offer to purchase legacy assets at some price $p_g$. The government purchase price $p_g$ is initially exogenous (at level above $I$); we later discuss how it may be chosen optimally in Section \ref{sec:wel} in light of the public cost of a bailout.  The timeline of our full game is depicted in Figure \ref{fig:timeline2}.

\begin{figure} [htp]
	\centering
	\begin{tikzpicture}[
	scale = 0.8,
	1p line/.style={thick, blue},
	2p line/.style={thick, darkred},
	reference/.style = {dashed, thick},
	axis/.style={very thick},
	move/.style={thick, ->, shorten <=2pt, shorten >=2pt},
	comment/.style={font=\scriptsize\sffamily}
	]
	
	\draw[move] (0, 0) -- (20, 0) node[comment, right, text width = 1cm] at (20, 0) {$t = 1$};
	
	\draw[move] (0, -4) -- (20, -4) node[comment, right, text width = 1cm] at (20, -4) {$t = 2$};

	\filldraw (1, 0) circle (3pt) node[comment, above, text width = 2.5cm, yshift = 3pt] at (1, 0) {Firms privately learn the value of their legacy assets.};
	
	\filldraw (6, 0) circle (3pt) node[comment, below, text width = 2.5cm, yshift = -3pt] at (6, 0) {Government offers to buy one unit of the asset at $p_g$.};
	
	\filldraw (11, 0) circle (3pt) node[comment, above, text width = 2.75cm, yshift = 3pt] at (11, 0) {Buyers in the market make offers.};
	
	\filldraw (16, 0) circle (3pt) node[comment, below, text width = 2.75cm, yshift = -3pt] at (16, 0) {Firms accept either a government offer, a market offer, or none. Those who sell fund the project.};
	

	
	
	\filldraw (3, -4) circle (3pt) node[comment, above, text width = 2.75cm, yshift = 3pt] at (3, -4) {Buyers  make offers for $t=2$ legacy assets.};
	
	\filldraw (10, -4) circle (3pt) node[comment, below, text width = 2.75cm, yshift = -3pt] at (10, -4) {Firms either sell to the buyers or hold out. Those who sell fund their projects.};
	
	\filldraw (17, -4) circle (3pt) node[comment, above, text width = 2.75cm, yshift = 3pt] at (17, -4) {Project returns for both periods are realized.};
	
	\end{tikzpicture}
	\caption{Timeline for the two-period model}\label{fig:timeline2}
\end{figure}

 To focus our attention on the main issue---namely, bailout stigma---we make several simplifying assumptions.

 {First, just like \cite{tirole2012overcoming}, we assume that the limited pledgeability of the project inhibits direct financing.  This means that  the sale of legacy assets  is the only means of funding the project for firms.  In the same vein, we consider the government's purchase of legacy assets as the only means of government bailout.\footnote{In fact, the asset purchase can be interpreted as a loan collateralized by the associated asset.  Due to the non-pledgeability, the asset must be seized and liquidated, which  makes the loan qualitatively equivalent to the asset purchase.} This is primarily a simplifying assumption. As shown by \cite{philippon2012optimal}, the main thrust of the analysis extends to the case in which the project can be pledged along with legacy assets as collateral to obtain financing.\footnote{Their insight appears to apply to our context, which suggests that debt contracts would be optimal in our context as well.  Since the stigma issue is separable from the issue of contract form, we abstract from it in the current paper. } From this broader perspective, adverse selection with respect to legacy assets must be interpreted as pertaining to their values as collateral required for financing; and our results can be translated into this broader context naturally. } 

Second, {to simplify the analysis, we assume that the return from $t = 1$ project is not pledgeable, and does not accrue until each firm needs to fund its $t = 2$ project. This implies that the firm is not able to   use the return from $t=1$ project to finance the $t = 2$ project.} 


Third, as is standard in the literature, we assume that asset buyers in the market are competitive and make purchase offers in Bertrand fashion.  Specifically, we require that an equilibrium admit no  deviation offer that a positive measure of firms have strict incentives to accept and at least breaks even for the deviating buyer.\footnote{This requirement is only slightly stronger than the standard equilibrium assumption, and can be seen as a refinement that bolsters the credibility of a selected equilibrium.} In particular, this means that in equilibrium all buyers must break even, since otherwise, there will be a weakly profitable deviating offer that will be attractive for a positive measure of firms.   Buyers live for one period and make offers that would break even in expectation.   Importantly, buyers in $t = 2$ can make rational inference about firms' types from their observable behavior in $t=1$, in particular with their acceptance/rejection of a bailout offer.


Fourth, we assume that the sale of assets to the market in $t = 1$ is private and therefore not revealed to buyers in $t=2$. This implies that buyers in the $t=2$ market cannot distinguish between those that sold in the $t=1$ market and those that did not.    Again the primary reason for this assumption is to simplify the analysis by shutting off  channels of dynamic information revelation. However, this assumption is well justified given that many important financial and real assets are sold privately over the counter.   The main thrust of our results extends to the case of observable sale, as shown by the working paper version of our paper (see \cite{che2018bailout}).  Moreover, this assumption makes a comparison with \cite{tirole2012overcoming} transparent, which helps to isolate the effect of stigma.

Fifth, a government bailout is available only in $t=1$.  This is consistent with the observed practice: governments refrain from engaging in long-term bailouts and from complete ``nationalization'' of distressed firms (which would be equivalent to purchasing two units of the asset in our model). Our goal is to study the reputational consequence of accepting or rejecting a bailout, which can be studied most effectively when no bailout is available in the second period.\footnote{Like \cite{tirole2012overcoming}, our focus is on market freezes and the government's role to alleviate them. This focus means that we abstract from other relevant issues such as the moral hazard associated with firms' excessive risk taking or shirking in anticipation of bailouts.  These moral hazard problems suggests the desirability of the government committing not to bail out firms that are {\it ex post} in need of rescue. This problem  and the associated time inconsistency problem are discussed in \cite{green}, \cite{chari2016}, and \cite{keister}. Since firms make no effort or risk taking decisions in our model, these issues do not arise in our model, thus allowing us to focus on the market freeze problem.}

 Finally, we focus on ``transparent'' bailouts; namely, the $t=2$ market observes the identities of firms that have accepted bailouts.   Not only do transparent bailouts  highlight the stigma effect most clearly, but they are also important from a practical perspective.  While secret bailouts may address the stigma problem, secrecy is often difficult to achieve in practice.  Nevertheless, one may consider suppressing information about bailouts totally or partially in the spirit of information design (\citet{bergemann}). Our analysis in Section \ref{sec:wel} encompasses this general information design perspective, and shows that a transparent bailout is without loss given the selection of equilibrium.

\section{Preliminary Analysis}	\label{sec:benchmark}

Before proceeding, we study several benchmarks. They will facilitate comparison with, and provide a context for, our main results, which will follow in  Section \ref{sec:gov}.

\subsection{ {\it  Laissez-faire} without Government Bailout}

We first consider the benchmark without a bailout. The timeline  is the same as Figure \ref{fig:timeline2}, except that the government's bailout is absent.  Since the  sale of the assets is private and not publicly revealed, there is no linkage between the market outcomes across two periods. Thus, the game reduces to a  one-period game (repeated twice) whose equilibrium coincides with that of Tirole's game without bailout.

Fix any period. The equilibrium outcome is understood best as a form of Akerlof's lemons problem, which is depicted in Figure \ref{fig:akerlof}. 
\begin{figure} [bht] \caption{Determination of the cutoff type $\t_0$ in the {\it laissez-faire} economy} \label{fig:akerlof}
	\centering
	\begin{tikzpicture}[
	scale = 7.5,
	comment/.style={font=\scriptsize\sffamily}
	]
	\draw[->, thin] (0, 0) -- (1, 0) node[comment, left] at (0, 0) {$0$} node[comment, right] at (1, 0) {marginal type $\hat \t$}; \draw[->, thin] (0, -.3) -- (0, .7) node[comment, left] at (0, .7) {Price};
	
	\draw[red!75!black, thick] (0, 0)..controls(.1, .08)and(.25,.17)..(.3,.20)..controls(.475,.285)and(.525,.325)..(0.65,0.4)..controls(.7,.43)and(.85,.48)..(.9, 0.5)..controls(.93,.505)and(.97,.525)..(1, .53) node[comment, right, text width = 2.5cm] at (1, .525) {Average Benefit Curve:  $\E[\t | \t \le \hat \t]$};
	
	
	\draw[green!45!black, thick] (0, -.25) -- (.95, 0.7) node[comment, right, text width = 2cm] at (.95, .7) {Supply Curve:  $\hat\t - S$}; 
	
	\draw[thick, dotted] (.65, .4) -- (.65, 0)  node[comment, below] at (.65, 0) {$\theta_0$};
	
	\draw[thick, dotted] (.65, .4) -- (0, .4)  node[comment, left] at (0, .4) {$p_0$};
	
	
	
	\node[comment, left] at (0, -.25) {$-S$};
	
	\end{tikzpicture}
\end{figure}
The figure plots two curves both as functions of the marginal type  of firm $\hat \t$  selling  to the market.  The marginal type $\hat \t$ effectively represents the ``quantity'' sold since types $\t$ sell if and only if $\t\le \hat \t$ in equilibrium.\footnote{This feature  follows from the single-crossing property:  if a type-$\t$ firm sells, then type-$\t'<\t$ firm strictly prefers to sell. The quantity sold is thus $F(\hat \t)$ which corresponds  to $\hat \t$ in one-to-one manner.}  The marginal type faces $\hat\t-S$ as the opportunity cost of selling: by selling the firm loses the asset of value $\hat \t$ but gains the net surplus $S$. Since the marginal type $\hat \t$ must be indifferent to selling in equilibrium, we have $p=\hat\t-S$, giving rise to the {\it supply curve}. 
   Meanwhile, buyers of asset quality $\t\le \hat \t$ enjoy benefit   $\E[\t|\t\le \hat \t]$ on average.   Bertrand competition among buyers means that  average benefit must equal price in equilibrium, giving rise to the {\it average benefit curve}.

	Clearly, supply and average benefit curves must intersect at the equilibrium marginal type $\hat \t=\t_0$, where $\t_0$ satisfies
\begin{equation}
\label{eq:theta0}
\t_0-S=\E[\t|\t\le \t_0]=:p_0.
\end{equation}
The log-concavity assumption means that the average benefit curve always crosses the supply curve from above. Hence, there is a unique threshold $\t_0$ satisfying this requirement, which induces a unique equilibrium.\footnote{It is routine to check that if $f$ is log-concave ($  { \partial^2 \log f(\t)}/{\partial \t^2}<0$ for all $\t$), then there is a unique $\t$ satisfying (\ref{eq:theta0}); see \cite{tirole2012overcoming}.}   Finally, for trade to occur in equilibrium,  price $p_0$ must  be at least $I$, or else there are no gains from trade.  

Figure \ref{fig:no_bailout} summarizes the equilibrium configuration.\footnote{As mentioned, buyers cannot update their information since the market transactions are private. If market transactions were observable, then trading decisions become dynamic, which makes analysis complicated; see \cite{che2018bailout}. }  
\begin{figure} [htp]
	\centering
	\begin{tikzpicture}[
	scale = 0.6,
	1p line/.style={thick, blue},
	2p line/.style={thick, darkred},
	reference/.style = {dashed, very thick},
	axis/.style={very thick},
	move/.style={very thick, ->, dashed, shorten <=2pt, shorten >=2pt},
	comment/.style={font=\footnotesize\sffamily}
	]
	\fill[fill = blue, opacity = 0.5] (0, 1-0.15) -- (10*1.15, 1-0.15) -- (10*1.15, -1 +0.15) -- (0, -1 +0.15);	
	\fill[fill = blue, opacity = 0.5] (0, -1-0.15) -- (10*1.15, -1-0.15) -- (10*1.15, -3 +0.15) -- (0, -3 +0.15);
	
	\draw[reference] (10*1.15, 1) -- (10*1.15, -3);
	
	\draw[axis] (10*1.15, .2) -- (10*1.15, -.2) node[comment, below, xshift = 22.25pt] at (10*1.15, -3) {$\t_0 = p_0 + S$};	
	
	\draw[axis] (0, 0) -- (20*1.15, 0) node[comment, left] at (0, 0) {$t = 1$}; \draw[axis] (0, .2) -- (0, -.2); \draw[axis] (20*1.15, .2) -- (20*1.15, -.2);
	\node[comment, below] at (0, -3) {$\t=0$};
	\node[comment, below] at (20*1.15, -3) {$\t=1$};
	
	\draw[axis] (0, -2) -- (20*1.15, -2) node[comment, left] at (0, -2) {$t = 2$}; \draw[axis] (0, -1.8) -- (0, -2.2); \draw[axis] (20*1.15, -1.8) -- (20*1.15, -2.2);
	
	\node[comment, above, blue] at (5, 1) {$p_0 = \E[ \t | \t \le \t_0]$};
	
	
	\end{tikzpicture}
	\caption{No bailout equilibrium} \label{fig:no_bailout}
\end{figure}


Adverse selection means that the above outcome is typically inefficient.  Specifically, if  $S< 1-\E[\t]$, then  $\t_0<1$, so not all firms sell and finance their projects.  It is also possible that $\t_0=0$, in which case the market freezes completely.  To focus on the nontrivial case, we assume  $\t_0<1$. For expositional ease, it is also convenient to focus on the partial freeze case ($\t_0>0$) in what follows. We will discuss the full freeze case ($\t_0=0$) later in Remark \ref{rem:full-freeze}. 

\begin{example} Consider the uniform case, that is, $F(\t) = \t$. In this case, the equilibrium price $p_0$ is determined by $p_0 = \E[ \t | \t \le \t_0] =  {\t_0}/{2}$. Suppose $\t_0 \in (0, 1)$. Then, from the indifference condition $\t_0 = p_0 + S = \frac{\t_0}{2} + S$, the cutoff type $\t_0$ is uniquely determined as $\t_0 = 2S$, and $p_0=S$, if $S \in \left[I, 1/2 \right)$.  If $S<I$, then the equilibrium price $p_0=S$ cannot fund the project, so the market fully freezes, and hence $\t_0 = 0$. If $S \ge 1 - \E[\t] = 1/2$, then $\t \le \E[\t ] + S$ for all $\t \in [0, 1]$, and therefore, $\t_0 = 1$.
\end{example}

\subsection{Bailout without Stigma:  One-Period Model}

We next consider another benchmark, the one-period bailout model by \cite{tirole2012overcoming} in which  the government offers to purchase   assets at price $p_g$ above the {\it laissez-faire} price $p_0$ before the market opens.  
Specifically, the timeline simply comprises $t=1$ in Figure \ref{fig:timeline2}.\footnote{An astute reader will notice that this timeline differs slightly from that considered by \cite{tirole2012overcoming}, where the market opens after firms have decided on the government offer.  We adopt the current timeline since it is arguably more realistic, and also it permits equilibrium existence more broadly for our two-period extension.  For the one-period version, the difference is immaterial, since the equilibrium under the current timeline is payoff-equivalent to Tirole's equilibrium for all players involved.  In addition, we do not invoke an equilibrium refinement adopted in \cite{tirole2012overcoming}, as the central feature of the equilibrium holds irrespective of the refinement.  See Remark \ref{rem:refinement}.}  Since there is no consequence of accepting a bailout from the government in this one-period model, there is no bailout stigma---at least in the sense we will capture in our two-period model later.\footnote{Note that \cite{tirole2012overcoming} and \cite{philippon2012optimal} do recognize ``stigma'' associated with the types of firms that accept a bailout, but they do not study its effect on the subsequent game nor on the initial decision to accept the bailout, the dual focuses of the current paper.}  Thus, this benchmark will help to identify the role of bailout stigma later in our main analysis.

Perfect Bayesian equilibrium in this game, which we simply refer to hereafter as an equilibrium, is characterized as follows.  Let $\mu_g$ and $\mu_m$ denote the fractions of types $\t\le p_g+S$ that sell to the government and to the market, respectively, where $\mu_g+\mu_m=1$.  We first argue that $\mu_g>0$.  If no firm accepts the government offer,  then the {\it laissez-faire} equilibrium will prevail, with marginal type $\t_0$ and equilibrium price $p_0=\E[\t|\t<\t_0]$.  Since $p_g>p_0$, however, firms will deviate to accept the government offer, a contradiction.

Next, suppose  $\mu_m>0$, so the market is active in equilibrium.\footnote{Under our timeline, the market may not be active in equilibrium.  To see how such an equilibrium can be supported, suppose a buyer  deviates and offers a price $p'>p_g$.  Since firms have not yet accepted the government offer by then (given our timeline), all types $\t<p'$ would accept the deviation offer, and the deviating buyer will suffer a loss since $p'>p_g>p_0$.}  This implies that the market price $p_m$ must equal the government price $p_g$, or else a lower offer will not be accepted.  Given this, firms will sell (either to the government or to the market) if and only if $\t<p_g+S$.  Let $\bar\t_g$ and $\bar\t_m$ denote the  average values of assets  sold to the government and the market, respectively ($\bar\t_m$ can be arbitrary in case $\mu_m=0$). Clearly, we must have 
\begin{equation}\label{eq:feas}
\mu_g \bar\t_g+\mu_m \bar\t_m =\E[\t|\t\le p_g+S].
\end{equation}
Further, since market buyers must break even when $\mu_m>0$, we have $p_m = \bar\t_m$. Given $p_g=p_m$, this in turn implies $\bar\t_m=p_g$. 

The central feature of \cite{tirole2012overcoming} follows from these observations:

\begin{theorem} [dregs-skimming] \label{thm:tirole} \

\begin{description}

		\item[(i)] In any equilibrium of the one-period benchmark with government offer  $p_g>\max\{p_0, I\}$, firms sell assets (either to the government or to the market) at price $p_g$   if and only if $\t<p_g+S$.  Since $p_g>p_0$,  more firms finance their projects than without the government intervention.
		
		\item [(ii)]If $\mu_m>0$ so that the market is active, then we must have  $\bar\t_g<\bar\t_m$; that is, on average lower value assets are sold to the government than to the market.

\end{description}  
	   
\end{theorem}

\begin{proof} See Appendix \ref{app:thm1}.  
\end{proof}
 
 Figure \ref{fig:bm_bailout} illustrates the outcomes with and without government bailout.  By offering a higher price $p_g$ than the {\it laissez-faire} price $p_0$, the government  does indeed take out relatively low-value assets, which in turn improves the perception of the assets sold to the market and thus alleviates adverse selection.

\begin{figure} [htp]
	\centering
	\begin{tikzpicture}[
	scale = 0.6,
	1p line/.style={thick, blue},
	2p line/.style={thick, darkred},
	reference/.style = {dashed, very thick},
	axis/.style={very thick},
	move/.style={very thick, ->, dashed, shorten <=2pt, shorten >=2pt},
	comment/.style={font=\footnotesize\sffamily},
	comment2/.style={font=\small}
	]
		
	\fill[fill = blue, opacity = 0.5] (0, 1 -0.15) -- (10*1.15, 1 -0.15) -- (10*1.15, -1 +0.15) -- (0, -1 +0.15);

	\fill[fill = blue, opacity = 0.5] (0, 1 -0.15 -0.15 -2 -3.5) -- (1*1.15, 1 -0.15 -0.15 -2 -3.5) -- (1.5*1.15, 1 -0.15 -0.125 -2 -3.5) -- (2.25*1.15, 1 - 0.15 - 0.375 -2 -3.5) -- (3.15*1.15, 1 -0.15 -0.15 -2 -3.5) -- (4*1.15, 1 -0.15 -0.65 -2 -3.5) -- (4.85*1.15, 1 -0.15 -0.375 -2 -3.5) -- (5.25*1.15, 1 -0.15 -0.85 -2 -3.5) -- (6.1*1.15, 1 -0.15 -0.65 -2 -3.5) -- (6.55*1.15, 1 -0.15 -0.95 -2 -3.5) -- (7.15*1.15, 1 - 0.15 -0.8 -2 -3.5) -- (8.25*1.15, 1 -0.15 -0.65 -2 - 3.5) -- (9.5*1.15, 1 -0.15 -0.9 -2 -3.5) -- (10.5*1.15, 1 -0.15 -1.1 -2 -3.5) -- (11.5*1.15, 1 - 0.15 -0.95 -2 -3.5) -- (13*1.15, -1 +0.15 + 0.2 -2 -3.5) -- (13*1.15, 1 -0.15 -2 -3.5) -- (0, 1 - 0.15 -2 -3.5);
	\fill[fill = red, opacity = 0.5] (0, 1 -0.15 -0.15 -2 -3.5) -- (1*1.15, 1 -0.15 -0.15 -2 -3.5) -- (1.5*1.15, 1 -0.15 -0.125 -2 -3.5) -- (2.25*1.15, 1 - 0.15 - 0.375 -2 -3.5) -- (3.15*1.15, 1 -0.15 -0.15 -2 -3.5) -- (4*1.15, 1 -0.15 -0.65 -2 -3.5) -- (4.85*1.15, 1 -0.15 -0.375 -2 -3.5) -- (5.25*1.15, 1 -0.15 -0.85 -2 -3.5) -- (6.1*1.15, 1 -0.15 -0.65 -2 -3.5) -- (6.55*1.15, 1 -0.15 -0.95 -2 -3.5) -- (7.15*1.15, 1 - 0.15 -0.8 -2 -3.5) -- (8.25*1.15, 1 -0.15 -0.65 -2 - 3.5) -- (9.5*1.15, 1 -0.15 -0.9 -2 -3.5) -- (10.5*1.15, 1 -0.15 -1.1 -2 -3.5) -- (11.5*1.15, 1 - 0.15 -0.95 -2 -3.5) -- (13*1.15, -1 +0.15 + 0.2 -2 -3.5) -- (13*1.15, -1 +0.15 + 0.2 -2 -3.5) -- (13*1.15, -1 +0.15 -2 -3.5) -- (0, -1 + 0.15 -2 -3.5);
	
	\draw[reference] (10*1.15, 1) -- (10*1.15, -1) node[comment, below] at (10*1.15, -1) {$\t_0$};
	\draw[reference] (13*1.15, -1 -3.5) -- (13*1.15, -3 -3.5) node[comment, below, xshift = 22.25pt] at (13*1.15, -3 -3.5) {$\hat\t_g = p_g + S$};
	
	
	\draw[axis] (0, 0) -- (20*1.15, 0); \draw[axis] (0, .2) -- (0, -.2); \draw[axis] (20*1.15, .2) -- (20*1.15, -.2);
	\node[comment, below] at (0, -3 -3.5) {$\t=0$};
	\node[comment, below] at (20*1.15, -3 -3.5) {$\t=1$};
	
	\draw[axis] (0, -2 -3.5) -- (20*1.15, -2 -3.5); 
	
	\draw[axis] (0, -1.8 -3.5) -- (0, -2.2 -3.5); \draw[axis] (20*1.15, -1.8 -3.5) -- (20*1.15, -2.2 -3.5);
	
	
	\node[comment] at (10*1.15, 2) {Equilibrium without bailout};
	\node[comment] at (10*1.15,  -3.5) {Equilibrium with bailout};	
	
	
	\node[comment2, right, text width = 16.25cm, inner sep= 5pt] at (0-3.25, -2.5-3 -3.5) {Note: the types selling to the market are depicted by blue and the types selling to the government are depicted by red.};
	
	\end{tikzpicture}
	\caption{Effects of bailout in the one-period benchmark}	  \label{fig:bm_bailout}

\end{figure}

More importantly for our purpose, assets are sold to the government at the same price as they are sold to the market.  This reflects the absence of stigma associated with accepting a bailout.  Plainly, in the one-period problem, firms that accept the bailout can do so without fear of consequences, since the game ends immediately after the bailout.

\begin{remark} [The role of the equilibrium refinement in \cite{tirole2012overcoming}]\label{rem:refinement}  \rm To obtain the ``dregs-skimming'' role of bailout,  \cite{tirole2012overcoming} invokes an equilibrium refinement---that the market sale collapses with an arbitrarily small probability.  Given the single crossing property implicit in the firms' payoffs, this refinement effectively ``forces'' the equilibrium to have the dregs-skimming feature: namely, there exists $\tilde\t\in (0,p_g+S)$ such that types $\t\le \tilde \t$ all sell to the government and types $\t\in (\tilde\t, p_g+S]$ all sell to the market.\footnote{Suppose a market sale is subject to probability $\epsilon>0$ of cancellation. If a type-$\t$ firm prefers to sell to the government, then $p_g+S\ge (1-\epsilon)(p_m+S)+ \epsilon \t$, where $p_m$ is the equilibrium market price. This means that all types $\t'<\t$ must strictly prefer to sell to the government.} Consequently, this refinement makes unclear whether dregs-skimming is an artifact of the refinement or something more fundamental.  By not invoking the refinement,   Theorem \ref{thm:tirole}   proves that dregs-skimming is fundamental (and not  driven by the refinement).  Without the refinement, however, there are multiple equilibria that differ in terms of $\mu_g$ and the value $\overline\t_g$, but every such equilibrium exhibits the ``dregs-skimming'' feature. \end{remark}

\subsection{Secret Bailout}

In order to identify the effects of bailout stigma, we need to understand what would happen if the policy maker could  eliminate the stigma altogether.  Imagine that the policy maker ``completely and successfully'' conceals the identities of the firms that accept the government offer.  This kind of secrecy would be an important feature of the bailout policy and is worth studying in its own right, precisely because of the issue of stigma.\footnote{{\citet{gorton2020fighting} supports such a policy. During crises, debt contracts lose ``information insensitivity'' as investors scrutinize the downside risk of underlying collaterals, leading to an adverse selection. They argue that withholding information on whether borrowers borrow from discount windows of central banks can make debtors less information sensitive and alleviate adverse selection.  As will be seen, secrecy has a more nuanced effect in our model.}}  
Nevertheless, we view secrecy primarily as a benchmark against which transparent bailouts are compared, given our premise that ``complete'' secrecy has been so far difficult to achieve despite many concerted efforts.\footnote{The identities of banks borrowing from the discount window facilities (DW) are occasionally leaked to either the news media or the market participants through a number of channels. First, despite the apparent secrecy attached to DW, the access to DW by borrowing firms has been detected by news media (\citet{armantier2015discount}, \citet{berry2012bloomberg}). For instance, the Financial Times reported the news that Deutsche Bank had borrowed from DW one day ago (see ``Fed fails to calm money markets,'' \emph{The Financial Times}, August 20, 2007). Second, the market participants can identify DW borrowers from these banks' market activities or the information released by the Fed. On its weekly report, the Fed discloses whether there is an increase in aggregate DW borrowing. In addition, financial institutions can observe whether a bank did not borrow or lend at the federal funds market at that time. Combining all the information, one can easily identify a DW borrower (\citet{haltom2011federal}).}

The equilibrium under complete secrecy is very easy to analyze.  Since, under secrecy, neither sales to the market nor sales to the government are observed, the former by assumption and the latter by secrecy, firms need not worry about the signaling consequences of their $t=1$ actions.  Hence, the equilibrium in $t=1$ coincides with that of Theorem \ref{thm:tirole}.  Given no informational leakage,  the outcome of $t=2$ coincides with the no-intervention benchmark.

\begin{theorem}[Secret bailouts] \label{thm:secrecy}  Suppose the government offers to purchase assets at $p_g>p_0$ with full secrecy.  Then, in equilibrium, firms accept the government offer in $t = 1$ if and only if $\t<p_g+S$. In $t=2$, firms sell assets to the market at price $p_0$ if and only if $\t<\t_0$.
\end{theorem}

\begin{figure} [htp]
	\centering
	\begin{tikzpicture}[
	scale = 0.6,
	1p line/.style={thick, blue},
	2p line/.style={thick, darkred},
	reference/.style = {dashed, very thick},
	axis/.style={very thick},
	move/.style={very thick, ->, dashed, shorten <=2pt, shorten >=2pt},
	comment/.style={font=\footnotesize\sffamily},
	comment2/.style={font=\small}
	]
	
	\fill[fill = blue, opacity = 0.5] (0, 1 -0.15 -0.15 -0 -0) -- (1*1.15, 1 -0.15 -0.15 -0 -0) -- (1.5*1.15, 1 -0.15 -0.125 -0 -0) -- (2.25*1.15, 1 - 0.15 - 0.375 -0 -0) -- (3.15*1.15, 1 -0.15 -0.15 -0 -0) -- (4*1.15, 1 -0.15 -0.65 -0 -0) -- (4.85*1.15, 1 -0.15 -0.375 -0 -0) -- (5.25*1.15, 1 -0.15 -0.85 -0 -0) -- (6.1*1.15, 1 -0.15 -0.65 -0 -0) -- (6.55*1.15, 1 -0.15 -0.95 -0 -0) -- (7.15*1.15, 1 - 0.15 -0.8 -0 -0) -- (8.25*1.15, 1 -0.15 -0.65 -0 -0) -- (9.5*1.15, 1 -0.15 -0.9 -0 -0) -- (10.5*1.15, 1 -0.15 -1.1 -0 -0) -- (11.5*1.15, 1 - 0.15 -0.95 -0 -0) -- (13*1.15, -1 +0.15 + 0.2 -0 -0) -- (13*1.15, 1 -0.15 -0 -0) -- (0, 1 - 0.15 -0 -0);
	\fill[fill = red, opacity = 0.5] (0, 1 -0.15 -0.15 -0 -0) -- (1*1.15, 1 -0.15 -0.15 -0 -0) -- (1.5*1.15, 1 -0.15 -0.125 -0 -0) -- (2.25*1.15, 1 - 0.15 - 0.375 -0 -0) -- (3.15*1.15, 1 -0.15 -0.15 -0 -0) -- (4*1.15, 1 -0.15 -0.65 -0 -0) -- (4.85*1.15, 1 -0.15 -0.375 -0 -0) -- (5.25*1.15, 1 -0.15 -0.85 -0 -0) -- (6.1*1.15, 1 -0.15 -0.65 -0 -0) -- (6.55*1.15, 1 -0.15 -0.95 -0 -0) -- (7.15*1.15, 1 - 0.15 -0.8 -0 -0) -- (8.25*1.15, 1 -0.15 -0.65 -0 -0) -- (9.5*1.15, 1 -0.15 -0.9 -0 -0) -- (10.5*1.15, 1 -0.15 -1.1 -0 -0) -- (11.5*1.15, 1 - 0.15 -0.95 -0 -0) -- (13*1.15, -1 +0.15 + 0.2 -0 -0) -- (13*1.15, -1 +0.15 + 0.2 -0 -0) -- (13*1.15, -1 +0.15 -0 -0) -- (0, -1 + 0.15 -0 -0);
	
	\fill[fill = blue, opacity = 0.5] (0, -1-0.15) -- (10*1.15, -1-0.15) -- (10*1.15, -3 +0.15) -- (0, -3 +0.15);
	
	\draw[reference] (10*1.15, -1) -- (10*1.15, -3);
	\draw[reference] (13*1.15, 1) -- (13*1.15, -3) node[comment, above, xshift = 11.5pt] at (13*1.15, 1) {$p_g + S$};
	
	\node[comment, below, xshift = 22.25pt] at (10*1.15, -3) {$\t_0 = p_0 + S$};	
	
	\draw[axis] (0, 0) -- (20*1.15, 0) node[comment, left] at (0, 0) {$t = 1$}; \draw[axis] (0, .2) -- (0, -.2); \draw[axis] (20*1.15, .2) -- (20*1.15, -.2);
	\node[comment, below] at (0, -3) {$\t=0$};
	\node[comment, below] at (20*1.15, -3) {$\t=1$};
	
	\draw[axis] (0, -2) -- (20*1.15, -2) node[comment, left] at (0, -2) {$t = 2$}; \draw[axis] (0, -1.8) -- (0, -2.2); \draw[axis] (20*1.15, -1.8) -- (20*1.15, -2.2);
	

	
	\node[comment2, right, text width = 16.25cm, inner sep= 5pt] at (0-3.25, -2.5-3) {Note: the types selling to the market are depicted by blue and the types selling to the government are depicted by red.};
	
	\end{tikzpicture}
	\caption{Equilibrium with secret bailouts} \label{fig:secret_bailout}

\end{figure}

\section{Government Bailout and Stigma} \label{sec:gov}

We now turn to the two-period game whose timeline is depicted in Figure \ref{fig:timeline2}. We continue to assume that the government offer is above the {\it laissez-faire} price: $p_g > p_0$. 
Otherwise, there is only a trivial equilibrium in which the {\it laissez-faire} outcome prevails, with no firms accepting the government offer.

We begin by analyzing the structure of a possible equilibrium. We focus on the equilibrium obtained in the limit as the relative weight $\delta<1$ for the $t = 2$ payoff approaches 1.\footnote{The assumption of $\delta < 1$ is meant to capture the fact that even though the reputational  consequence of accepting a bailout  may be important, its effect does not outweigh the direct payoff consequence of the decision, which appears to be first-order. Further, the reputational impact attached to bailouts does not usually persist after the financial crisis resolves. For instance, the total outstanding TARP bank funds, after the launch of TARP in October 2008, reached their peak at \$235.3 billion in February 2009, but sharply decreased to \$80.4 billion in January 2010 (go to \url{https://www.treasury.gov/initiatives/financial-stability/reports/Pages/TARP-Tracker.aspx} for more details). This observation indicates that the banks, after having joined TARP during 2007-2009 Great Recession, had little trouble securing funding through the market after the crisis was over. }

\begin{lemma} \label{lem:char} In any equilibrium with  $p_g>p_0$, there are three cutoffs $0<\hat \t\le \hat\t_g\le \t_2$ such that types $\t\le\hat \t$ sell assets in both periods, some measure of whom sell to the government and the remainder of whom sell to the market in $t=1$; types $\t\in (\hat \t, \hat\t_g]$ sell only to the government in $t=1$ but do not sell in $t=2$; types $\t\in (  \hat\t_g,\t_2]$  sell only in $t=2$; and types $\t>\t_2$ never sell their assets in either period.
\end{lemma}

The structure of an equilibrium is depicted in   Figure \ref{fig:eq-structure}.

\begin{figure} [htp]
	\centering
	\begin{tikzpicture}[
	scale = 0.6,
	1p line/.style={thick, blue},
	2p line/.style={thick, darkred},
	reference/.style = {dashed, very thick},
	axis/.style={very thick},
	move/.style={very thick, ->, dashed, shorten <=2pt, shorten >=2pt},
	comment/.style={font=\footnotesize\sffamily},
	comment2/.style={font=\small}
	]
	\fill[fill = aurometalsaurus, opacity = 0.5] (8*1.15, 1 -0.15) -- (11*1.15, 1 -0.15) -- (11*1.15, -1 +0.15) -- (8*1.15, -1+0.15);
	\fill[fill = aurometalsaurus, opacity = 0.5] (11*1.15, -1-0.15) -- (13*1.15, -1-0.15) -- (13*1.15, -3+0.15) -- (11*1.15, -3+0.15);	
	\fill[fill = aurometalsaurus, opacity = 0.5] (0, 1-0.15) -- (8*1.15, 1-0.15) -- (8*1.15, -1 +0.15) -- (0, -1 +0.15);	
	\fill[fill = aurometalsaurus, opacity = 0.5] (0, -1-0.15) -- (8*1.15, -1-0.15) -- (8*1.15, -3 +0.15) -- (0, -3 +0.15);
	
	\draw[reference] (8*1.15, 1) -- (8*1.15, -3);
	
	\draw[axis] (8*1.15, .2) -- (8*1.15, -.2) node[comment, below] at (8*1.15, -3) {$\hat\t$};	
	\draw[axis] (11*1.15, .2) -- (11*1.15, -.2) node[comment, below] at (11*1.15, -3) {$\hat\t_g$};
	\draw[axis] (13*1.15, .2-2) -- (13*1.15, -.2-2) node[comment, below] at (13*1.15, -3) {$\t_2$};	
	
	\draw[axis] (0, 0) -- (20*1.15, 0) node[comment, left] at (0, 0) {$t = 1$}; \draw[axis] (0, .2) -- (0, -.2); \draw[axis] (20*1.15, .2) -- (20*1.15, -.2);
	\node[comment, below] at (0, -3) {$\t=0$};
	\node[comment, below] at (20*1.15, -3) {$\t=1$};
	
	\draw[axis] (0, -2) -- (20*1.15, -2) node[comment, left] at (0, -2) {$t = 2$}; \draw[axis] (0, -1.8) -- (0, -2.2); \draw[axis] (20*1.15, -1.8) -- (20*1.15, -2.2);
	
	\draw[reference] (11*1.15, 1) -- (11*1.15, -3);
	\draw[reference] (13*1.15, 1) -- (13*1.15, -3);	
	
	\end{tikzpicture}
	\caption{General structure of equilibrium} \label{fig:eq-structure}
\end{figure}

Lemma \ref{lem:char} rests on several observations. First, firms' preferences satisfy the single-crossing property, implying that  a lower type has greater incentives to sell than a higher type in either period. This implies that the total number of units sold in equilibrium across the two periods must be non-increasing in $\t$.  Second, the fact that buyers (either the government or the market) never ration sellers means that the quantity traded for each firm must be either zero or one in each period.\footnote{ {The same logic implies that virtually no partial sales would arise even though we allow for them.  By the single crossing property of the payoff functions, partial sales, just like randomization, could occur in equilibrium only for a measure zero set of firm types, namely, the type that is indifferent between selling two or one units and the type that is indifferent between selling one unit and zero unit.}}  Third, an arbitrarily small discounting of the second-period payoff, along with the first two observations, implies that, among those that sell only in one period, early sellers are of lower types than late sellers. These observations give rise to the stated cutoff structure, as depicted in Figure \ref{fig:eq-structure}.  We omit the formal proof since it follows from a standard argument based on these observations.



In what follows, we limit attention to the case of $p_g<1-S$, namely, when the offer is not so high that all firms would accept the bailout.\footnote{Given the condition, we will have $\hat\t_g < 1$.  In case $p_g$ is higher so that all firms accept the bailout, the second-period would coincide with the {\it laissez-faire} outcome. 
 }  Then, Lemma \ref{lem:char} implies that there are only two possible types of equilibria: (a)  {\it short-lived stimulation}  equilibria and (b)  {\it delayed stimulation}  equilibria, depending on whether the stimulation effect of a bailout arises in $t=1$ or delayed to $t=2$.  More formally, these two types of equilibria correspond to the cases where $\hat\t_g =\t_2$ and $\hat\t_g<\t_2$, respectively, in the cutoff structure characterized in Lemma \ref{lem:char}.    In Online Appendix \ref{Asec:equili_types}, we formally show that these are the only possible types of equilibria.

\subsection{Short-lived Stimulation Equilibria (SSE)}	\label{sec:no_delay}

This type of equilibrium corresponds to the case where  $\hat{\t}_g=\t_2$ in Lemma \ref{lem:char}, and is depicted in Figure \ref{fig:no-delay}.  Importantly, the segment $[\hat\t_g, \t_2]$ of firms selling in $t=2$ in Lemma \ref{lem:char} (i.e., delayed selling) is of measure zero in this equilibrium.  Consequently, a bailout triggers an immediate increase in trade volume in $t=1$, but, as we will argue, this effect is short-lived.

\begin{figure} [htp]
	\centering
	\begin{tikzpicture}[
	scale = 0.6,
	1p line/.style={thick, blue},
	2p line/.style={thick, darkred},
	reference/.style = {dashed, very thick},
	axis/.style={very thick},
	move/.style={very thick, ->, dashed, shorten <=2pt, shorten >=2pt},
	comment/.style={font=\footnotesize\sffamily},
	comment2/.style={font=\small}
	]
	\fill[fill = red, opacity = 0.5] (8*1.15, 1 -0.15) -- (12*1.15, 1 -0.15) -- (12*1.15, -1 +0.15) -- (8*1.15, -1 +0.15);
	\fill[fill = blue, opacity = 0.3] (0, 1 -0.15 -0.15) -- (1*1.15, 1 -0.15 -0.15) -- (1.5*1.15, 1 -0.15 -0.125) -- (2.25*1.15, 1 - 0.15 - 0.375) -- (3.15*1.15, 1 -0.15 -0.15) -- (4*1.15, 1 -0.15 -0.65) -- (4.85*1.15, 1 -0.15 -0.375) -- (5.25*1.15, 1 -0.15 -0.85) -- (6.1*1.15, 1 -0.15 -0.65) -- (6.55*1.15, 1 -0.15 -0.95) -- (7.15*1.15, 1 - 0.15 -0.8) -- (8*1.15, -1 +0.15 + 0.2) -- (8*1.15, 1 -0.15) -- (0, 1 - 0.15);
	\fill[fill = red, opacity = 0.5] (0, 1 -0.15 -0.15) -- (1*1.15, 1 -0.15 -0.15) -- (1.5*1.15, 1 -0.15 -0.125) -- (2.25*1.15, 1 - 0.15 - 0.375) -- (3.15*1.15, 1 -0.15 -0.15) -- (4*1.15, 1 -0.15 -0.65) -- (4.85*1.15, 1 -0.15 -0.375) -- (5.25*1.15, 1 -0.15 -0.85) -- (6.1*1.15, 1 -0.15 -0.65) -- (6.55*1.15, 1 -0.15 -0.95) -- (7.15*1.15, 1 - 0.15 -0.8) -- (8*1.15, -1 +0.15 + 0.2) -- (8*1.15, -1 +0.15) -- (0, -1 + 0.15);
	
	\fill[fill = blue, opacity = 0.3] (0, 1 -0.15 -0.15 -2) -- (1*1.15, 1 -0.15 -0.15 -2) -- (1.5*1.15, 1 -0.15 -0.125 -2) -- (2.25*1.15, 1 - 0.15 - 0.375 -2) -- (3.15*1.15, 1 -0.15 -0.15 -2) -- (4*1.15, 1 -0.15 -0.65 -2) -- (4.85*1.15, 1 -0.15 -0.375 -2) -- (5.25*1.15, 1 -0.15 -0.85 -2) -- (6.1*1.15, 1 -0.15 -0.65 -2) -- (6.55*1.15, 1 -0.15 -0.95 -2) -- (7.15*1.15, 1 - 0.15 -0.8 -2) -- (8*1.15, -1 +0.15 + 0.2 -2) -- (8*1.15, 1 -0.15 -2) -- (0, 1 - 0.15 -2);
	\fill[fill = red, opacity = 0.5] (0, 1 -0.15 -0.15 -2) -- (1*1.15, 1 -0.15 -0.15 -2) -- (1.5*1.15, 1 -0.15 -0.125 -2) -- (2.25*1.15, 1 - 0.15 - 0.375 -2) -- (3.15*1.15, 1 -0.15 -0.15 -2) -- (4*1.15, 1 -0.15 -0.65 -2) -- (4.85*1.15, 1 -0.15 -0.375 -2) -- (5.25*1.15, 1 -0.15 -0.85 -2) -- (6.1*1.15, 1 -0.15 -0.65 -2) -- (6.55*1.15, 1 -0.15 -0.95 -2) -- (7.15*1.15, 1 - 0.15 -0.8 -2) -- (8*1.15, -1 +0.15 + 0.2 -2) -- (8*1.15, -1 +0.15 -2) -- (0, -1 + 0.15 -2);
	
	\draw[reference] (8*1.15, 1) -- (8*1.15, -3);
	
	\draw[axis] (8*1.15, .2) -- (8*1.15, -.2) node[comment, below] at (8*1.15, -3) {$\hat\t$};	\draw[axis] (12*1.15, .2) -- (12*1.15, -.2) node[comment, below] at (12*1.15, -3) {$\hat\t_g$};
	
	\draw[axis] (0, 0) -- (20*1.15, 0) node[comment, left] at (0, 0) {$t = 1$}; \draw[axis] (0, .2) -- (0, -.2); \draw[axis] (20*1.15, .2) -- (20*1.15, -.2);
	\node[comment, below] at (0, -3) {$\t=0$};
	\node[comment, below] at (20*1.15, -3) {$\t=1$};
	
	\draw[axis] (0, -2) -- (20*1.15, -2) node[comment, left] at (0, -2) {$t = 2$}; \draw[axis] (0, -1.8) -- (0, -2.2); \draw[axis] (20*1.15, -1.8) -- (20*1.15, -2.2);
	
	\draw[reference] (12*1.15, 1) -- (12*1.15, -3);
	\node[comment2, right, text width = 16.25cm, inner sep= 5pt] at (0-3.25, -2.5-3) {Note: the types selling to the market are depicted by blue and the types selling to the government are depicted by red.};
	
	\end{tikzpicture}
	\caption{Short-lived stimulation equilibrium} \label{fig:no-delay}
	
\end{figure}

 
 The SSE are characterized as follows:
 
 \begin{theorem} [\emph{Short-lived stimulation} equilibria] \label{thm:nodelay} 
 In an SSE under an public offer $p_g>p_0$, there exist $\hat\t<\t_0<  p_g + S=\hat\t_g$ such that the following holds:
 \begin{description}
 \item  [(i)] Types $\t\le \hat\t$ sell in both periods, and types $\t\in (\hat\t, \hat\t_g)$ sell only in $t=1$ to the government. 
 \item  [(ii)] Among the types $\t\le \hat\t$, the fraction $\mu_g>0$ with average quality $\bar\t_g$ accepts the bailout and the fraction $\mu_m = 1 - \mu_g>0$ with average quality $\bar\t_m$ sells to the market at $p_m =\bar\t_m$ in
 $t=1$, where $\bar\t_g < \bar\t_m < p_g$. 
 \item [(iii)]     $p_g< 2 \t_0 - \E[ \t | \t \le \t_0]$. 
 \end{description}
 
 \end{theorem}

\begin{proof}  See Appendix \ref{app:thm3}.
\end{proof}


We highlight three features of SSE.  
 First, bailout recipients suffer stigma.  In the $t=2$ market, bailout recipients are believed to be of  type $\bar\t_g$ (on average), while those that sell to the market in $t=1$ are believed to be of  type $\bar\t_m > \bar\t_g$. Thus, assets held by bailout recipients are sold at discount precisely equal to $\Delta=\bar\t_m-\bar\t_g$, since the market correctly infers the difference in their average asset values.   Of course, this stigma must be compensated in $t=1$, or else no firms would accept the bailout.  In particular, the government must pay more than the market does for the asset in $t=1$.  Since the government offer is fixed at $p_g$, this means that the market in $t=1$ must clear at price $p_g-\Delta$.  In other words, buyers demand a haircut $\Delta$ from firms selling to them for avoiding that stigma. So, the stigma leads to a commensurate mark-down of price for assets traded in $t=1$ market.\footnote{{While direct evidence for such a mark-down is not readily available, the phenomenon is similar to the finding in \citet{armantier2015discount} that borrowing from the Discount Window was cheaper than using the Term Auction Facility due to the stigma attached to the former.}  }   Since the market is competitive, buyers cannot earn positive profit, so what this simply means is that the average type $\bar\t_m$ of firms selling to the market must (endogenously) equal $p_g-\Delta$.

Second, dregs-skimming by government bailout---featured prominently in Tirole's model---does not occur here.  More specifically, there is a positive measure of firms $\t \in (\hat\t, \hat\t_g]$ that accept the government offer. The assets sold by these firms have strictly higher value than $\hat\t$, the highest value of asset sold to the market. This contrasts with the equilibrium in \cite{tirole2012overcoming}, in which the assets sold to the government are  worth strictly less than those sold to the private markets. 
Bailout stigma here creates the incentive for high-type firms to mitigate it or avoid it altogether.  Selling to the market in $t=1$ instead is one option, but it is subject to a mark-down of asset price by $\Delta$; effectively, bailout stigma has ``spread'' to market sellers in $t=1$.  Another way to avoid the collateral damage is to accept the bailout but simply withdraw from the $t=2$ market.  Indeed, types $(\hat \t, p_g+S]$ find it strictly profitable to accept the bailout but refuse to sell assets in $t=2$ to avoid the stigma.  The presence of these firms undercuts the government's effort to take out the most toxic assets and boost the market reputation of the remaining firms.  This has a long term effect, as we now turn to.

Third,  the government bailout crowds out the private markets in both periods, and makes stimulation short-lived. 
This is caused by  high-type firms'  withdrawal from the $t=2$ market to avoid stigma, which worsens the reputation of firms that \emph{do} participate in the $t=2$ market: they are effectively revealed to be of type $\bar\t_g$ on average. {The high-type firms' behavior allows the stigma to spread to all firms, as noted above, and ends up suppressing the market prices in $t=2$, including those that did not accept the bailout in $t=1$. Naturally, this (anticipated) effect on the asset prices feeds back into the asset market in $t=1$, dampening its price below $p_g$. This stands in sharp contrast to Tirole's static model, in which the government offer $p_g$ above the {\it laissez-faire} price $p_0$ raises the private market offer also above $p_0$, in fact exactly to $p_g$. Quite the contrary, here the government offer crowds out the markets in both periods, causing them to perform worse than even the {\it laissez-faire} benchmark.} In fact, this negative effect is so severe that the  $t=2$ market freezes more than if there were no bailout in $t=1$: only types $\t\le\hat\t$ sell in $t=2$, where importantly $\hat\t<\t_0$.  By comparison, all types $\t\le \t_0$ would have traded in $t=2$ in the absence of bailout (recall Figure \ref{fig:no_bailout}).  Clearly, transparency leads to a strict loss of trade.  Compared with a secret bailout, the volume of trade (and investment) induced under the transparent bailout is the same in $t=1$ but strictly smaller in $t=2$. 
   

The properties identified so far are necessary but not sufficient for SSE.   Specifically, for the existence of SSE, buyers targeting bailout recipients should not gain from raising their offers to attract the boycotters (i.e., types $\t\in [\hat\t,p_g+S]$), and the buyers targeting non-recipients should have no incentives to raise their offers to attract holdouts (i.e., types $\t>p_g+S$) together with the market sellers.  These conditions are formally stated and shown to be sufficient in Online Appendix \ref{Asec:SL}.

  These conditions are not easy to check, so it is difficult to establish the existence of the equilibrium (or its sufficient condition) in a simple manner.  Nevertheless, SSE exist for a range of bailout offers under many common distribution functions $F$.  For example, Figure \ref{fig:uni_ex}-(a) in Section \ref{sec:uniform}) shows (a continuum of) SSE when $F$ is uniform.

 On the other hand, SSE do not exist if $p_g$ is sufficiently high.  More precisely, as stated in  (iii) of Theorem \ref{thm:nodelay}, SSE disappear if $p_g\ge 2 \t_0 - \E[ \t | \t \le \t_0]$.  Roughly speaking, if $p_g$ is sufficiently high, accepting the government offer becomes very attractive, so the measure of firms selling to the market in $t=1$  vanishes.  This induces the buyers in $t=2$ market to target $t=1$ holdouts with a high  offer;  anticipating this, in $t=1$ high-type firms would deviate to hold out instead of selling to the government, undermining the SSE.
One implication of this condition is that the policy maker can  avoid triggering  undesirable SSE by making the bailout offer sufficiently generous---a point that will become clear as we now turn to delayed stimulation equilibria.

	
\subsection{Delayed Stimulation Equilibria (DSE)}	\label{sec:delay}

Delayed Stimulation Equilibria (DSE) have the structure that $\hat\t \le \hat{\t}_g < \t_2$ in the characterization in Lemma \ref{lem:char}, as illustrated in Figure \ref{fig:delay}. We call this \emph{delayed stimulation} equilibrium since much of the stimulation effect materializes in $t=2$.  In particular, the highest type that trades does so in $t=2$.  Types $\t<\hat{\t}_g$ act similarly as in the SSE:  nonnegative fractions $\mu_g$ and $\mu_m$ of types $\t\le \hat{\t}$ sell respectively to the government and to the market, and types $\t\in(\hat{\t}, \hat{\t}_g)$ sell only to the government, where $\hat{\t}\le\hat{\t}_g$.

\begin{figure} [htp]
	\centering
	\begin{tikzpicture}[
	scale = 0.6,
	1p line/.style={thick, blue},
	2p line/.style={thick, darkred},
	reference/.style = {dashed, very thick},
	axis/.style={very thick},
	move/.style={very thick, ->, dashed, shorten <=2pt, shorten >=2pt},
	comment/.style={font=\footnotesize\sffamily},
	comment2/.style={font=\small}
	]
	\fill[fill = red, opacity = 0.5] (9*1.15, 1 -0.15) -- (11.5*1.15, 1 -0.15) -- (11.5*1.15, -1 +0.15) -- (9*1.15, -1+0.15);
	\fill[fill = blue, opacity = 0.5] (11.5*1.15, -1 -0.15) -- (13*1.15, -1 -0.15) -- (13*1.15, -3+0.15) -- (11.5*1.15, -3+0.15);	
	\fill[fill = blue, opacity = 0.3] (0, 1-0.15) -- (9*1.15, 1-0.15) -- (9*1.15, -1 +0.15 + 1.25) -- (0, -1 +0.15 + 1.25);
	\fill[fill = red, opacity = 0.5] (0, -1 +0.15) -- (9*1.15, -1 +0.15) -- (9*1.15, -1 +0.15 + 1.25) -- (0, -1 +0.15 + 1.25);	
	\fill[fill = blue, opacity = 0.5] (0, 1-0.15 -2) -- (9*1.15, 1-0.15 -2) -- (9*1.15, -1 +0.15 + 1.25 -2) -- (0, -1 +0.15 + 1.25 -2);
	\fill[fill = red, opacity = 0.5] (0, -1 +0.15 -2) -- (9*1.15, -1 +0.15 -2) -- (9*1.15, -1 +0.15 + 1.25 -2) -- (0, -1 +0.15 + 1.25 -2);
	
	\draw[reference] (9*1.15, 1) -- (9*1.15, -3);
	
	\draw[axis] (9*1.15, .2) -- (9*1.15, -.2) node[comment, below, xshift = 10pt] at (9, -3) {$\hat\t = \t_0$};	
	\draw[axis] (11.5*1.15, .2) -- (11.5*1.15, -.2) node[comment, below] at (11.5*1.15, -3) {$\hat\t_g$};
	\draw[axis] (13*1.15, .2-2) -- (13*1.15, -.2-2) node[comment, below, xshift = 22.5pt, yshift = -2.7pt] at (13*1.15, -3) {$\t_2 = p_g + S$};	
	
	\node[comment, above] at (4.5*1.15, 1) {$\bar\t_g = \bar\t_m = \E[ \t | \t \le \t_0]$};
	
	\draw[axis] (0, 0) -- (20*1.15, 0) node[comment, left] at (0, 0) {$t = 1$}; \draw[axis] (0, .2) -- (0, -.2); \draw[axis] (20*1.15, .2) -- (20*1.15, -.2);
	\node[comment, below] at (0, -3) {$\t=0$};
	\node[comment, below] at (20*1.15, -3) {$\t=1$};
	
	\draw[axis] (0, -2) -- (20*1.15, -2) node[comment, left] at (0, -2) {$t = 2$}; \draw[axis] (0, -1.8) -- (0, -2.2); \draw[axis] (20*1.15, -1.8) -- (20*1.15, -2.2);
	
	\draw[reference] (11.5*1.15, 1) -- (11.5*1.15, -3);
	\draw[reference] (13*1.15, 1) -- (13*1.15, -3);

	\node[comment2, right, text width = 16.25cm, inner sep= 5pt] at (0-3.25, -2.5-3) {Note: the types selling to the market are depicted by blue and the types selling to the government are depicted by red.};
	
	\end{tikzpicture}
	\caption{Delayed stimulation equilibrium} \label{fig:delay}
\end{figure}

What makes this equilibrium possible is the incentive that  $t=2$ buyers have to offer a sufficiently high price to high-type firms holding out in $t=1$. Such an incentive was absent in  SSE due to a sizable fraction $\mu_m$ of low-type market sellers. These firms cannot be distinguished from high-type hold-out firms and, therefore, would inflict a loss to buyers if they were to raise offers to attract high-type hold-out firms. In DSE, the fraction $\mu_m$ of low-type market sellers is sufficiently small, especially when $p_g$ is large,  so that   $t=2$ buyers do have an incentive to attract high-type holdouts, unlike in SSE.

We now provide the characterization of DSE. 

 \begin{theorem} [\emph{Delayed stimulation} equilibria]  \label{thm:delay}  There is a DSE if $p_g \ge \E[\t| \t\in [\t_0, \gamma(\t_0)]]$, where  $\gamma(\t'):=\sup\{\t''\in [0,1]:\E[\t|\t\in [\t',\t'']]+S\ge \t''\}$ for each $ \t'$.\footnote{In words, $\gamma(\t')$ is the highest type $\t''$ such that types $[\t',\t'']$ can be induced to sell to the market at a break-even price equal to their conditional expectation.}  In the DSE,  $\hat\t=\t_0$, $\hat\t_g \in [\hat\t, p_g+S)$ and $\t_2=p_g+S$ such that the following holds.
 	\begin{description}
 	
	\item[(i)] Types $\t\le \t_0$ sell in both periods, a positive (possibly the entire) measure of which accept the bailout. 
 		
 	\item [(ii)] Among types $\t\le \t_0$, those that sell to the market in $t=1$ (if they exist) receive price $p_0$ in $t=1$ and $p_g$ in the $t=2$ market, and those that accept the bailout sell  assets at price $p_0$ in $t=2$.  Furthermore, these two groups of firms have the same average value of $\E[\t|\t\le\t_0]=p_0$.
 		
   	\item[(iii)]  Types $\t\in (\t_0,\hat{\t}_g]$ sell only in $t=1$ and to the government at price $p_g$.
 	
 	\item [(iv)]  Types $\t\in (\hat{\t}_g, p_g+S]$ sell only in $t=2$ at price $p_g$.  Higher-type firms never sell in either period.
	
	
 	\end{description}      
 \end{theorem}
 
 \begin{proof}  See Appendix \ref{app:thm4}.   \end{proof}

DSE are similar to SSE in two ways.  First, just as in SSE, firms suffer from accepting a bailout in the $t=2$ market. That is, the market in $t=2$ offers a strictly lower price to the bailout recipients than those that did not accept the bailouts, which include the firms that sold to the market in $t=1$ and those that held out.  Also as with SSE, this differential treatment of bailout recipients relative to initial market sellers in the $t = 2$ market can occur in equilibrium only if it is counterbalanced by the opposite treatment of these two groups in $t=1$.  Specifically, a form of the no-arbitrage condition must hold for the firms that sell in both periods: 
      $$p_g + p^g_2 = p_m + p^m_2,$$
where $p_m$ is the equilibrium market price for asset in $t=1$, and $p^g_2$ and $p^m_2$ are  the equilibrium $t=2$ market offers to the recipients and the non-recipients of the bailout, respectively.   {In short, those selling in both periods must receive the same total payment.}

Despite these similarities, there is a crucial difference, pertaining to the level of  $p^m_2$.  Recall that in SSE, high-type firms accept the bailout in $t=1$ but withdraw from the market; this means that those   participating in the $t=2$ market tend to have low value, which ``pushes down'' the $t=2$ market prices (both $p^g_2$ and $p^m_2$) to a level below $p_g$.  Hence, as captured by the above no-arbitrage equation and Theorem \ref{thm:nodelay}, in SSE, bailout stigma spreads to all firms selling in both periods, with the consequence that market prices for all assets fall below the {\it laissez-faire} level $p_0$ in both periods. 

By contrast, in DSE, high-type firms hold out in $t=1$ and sell to the $t=2$ market, and their participation in the $t=2$ market ``pushes up'' its price $p^m_2$ for the non-recipients of the bailout.  In fact, equilibrium dictates that $p^m_2$ must be at least $p_g$, or else the $t=1$ holdouts would rather sell to the government at $p_g$ in $t=1$ and withdraw from the $t=2$ market.  {In fact, $p^m_2$ must equal $p_g$, as stated in part (ii). To explain why, consider a simple case of DSE with $\mu_m = 0$; that is, every firm with $\t \le \hat\t$ sells to the government in $t = 1$.  (The proof treats all cases.)  First of all, we must have $p^m_2 \ge p_g$, or no firm would ever sell only in $t=2$ since it would be strictly better to sell to the government in $t=1$ and withdraw from the market in $t=2$.  Suppose next $p^m_2> p_g$.  Then, by the indifference of type $\hat\t$, we have
$$ \hat\t = p^g_2 + S - (p^m_2 - p_g) < p^g_2 + S.$$
Since $p^g_2 = \E[\t | \t \le \hat\t]$ (by  $t = 2$ buyers' break-even condition), we then have $\hat\t < \E[\t | \t \le \hat\t] + S$, or equivalently, $\hat\t < \t_0$. If a $t = 1$ buyer offers a price $p' \in (p^g_2, p_0)$, all types $\t \le \hat\t$ will sell to the deviating buyer because $p' + p^m_2 > p^g_2 + p_g$. Furthermore, the highest type $\t'$ that sells to the market over two periods will be determined by $\t' + p^m_2 + S = p' + p^m_2 + 2S$, or equivalently, $\t' = p' + S$. By definition of $\t_0$, however, we have $\E[ \t | \t \le \t'] - p' > 0$, so the deviating $t = 1$ buyer can make a profit, a contradiction.}

The high price for the non-recipients in $t = 2$ in turn incentivizes high-type firms to hold out, validating the premise of the DSE.  More interestingly,  the high price $p^m_2=p_g$ is not just enjoyed by the $t=1$ holdouts but also by the $t=1$ market sellers.  This feature keeps  the bailout stigma from  spreading to the $t=1$ market sellers.  Quite the contrary, their good reputation ``propped up'' by the $t=1$ holdouts  mitigates the stigma suffered by the bailout recipients.  This can be seen in the above no-arbitrage equation.  Substituting $p^m_2=p_g$ in that equation yields $p^g_2 = p_m$: the bailout recipients enjoy the same price in the $t=2$ market as the price that prevails in the $t=1$ market.  This is why, unlike SSE,  the {\it laissez-faire} price $p_0$ can be supported for both $t=1$ market as well as for the $t=2$ markets.\footnote{Several steps are needed to arrive at this conclusion.  First, since the buyers must break even, the observation $p^g_2 = p_m$ earlier implies that the average types of bailout recipients and $t=1$ market sellers must be identical and thus equal to $\E[\t |\t \le \hat\t]$.  This, together with the fact that the marginal type $\hat \t$ must be indifferent to selling in both periods versus selling only in one period gives rise to $\E[\t |\t \le \hat\t]+S=\hat \t$, the condition that defines the {\it laissez-faire} cutoff $\t_0$ in  (\ref{eq:theta0}). It now follows immediately that $p^g_2 = p_m=\E[\t |\t \le \t_0]= p_0$.}  In a nutshell, the high-type firms' holdout decision  keeps bailout stigma from freezing the markets. Although the government offer does not boost the private market offer as in \cite{tirole2012overcoming}'s static setup, it does not suppress the market offer, unlike in SSE.

	The above discussions lead to the following key observation.

\begin{corollary}  \label{cor1}  Each DSE is equivalent to an equilibrium under secret bailout (described in Theorem \ref{thm:secrecy}) in total volume of asset sales---and thus in total investments undertaken by firms.  The total volume of asset sales in any DSE exceeds that in SSE.
\end{corollary}

\begin{proof} See Appendix \ref{app:cor1}.
\end{proof}

One interesting, and perhaps surprising, implication of this result is that delays in the effect of a bailout should not be viewed as a policy failure, at least when compared to the secret bailout benchmark. Take a possible equilibrium where $\hat{\t}_g$ is very low; in fact, one can show that   $\hat{\t}_g=\t_0$ can be supported if $p_g$ is not too high  (again such an equilibrium arises due to a large mass of high type firms holding out in $t=1$). In such an equilibrium, it may appear at $t=1$ that  bailouts have no impact, since trading and investment activity have not changed after the government offer.  The impression would be that government purchases ``crowded out'' private purchases, since a positive measure of firms with $\t\le\t_0$ sell to the government at a higher price $p_g$ than $p_0$, the price they would have sold at in the absence of bailouts.  Indeed, in the wake of the Great Recession, such a sentiment prevailed following the apparent lack of response by banks to the first wave of stimulation policies.\footnote{In fact, as pointed out by \citet{bolton2009market}, LIBOR-OIS spreads did not decrease following the implementation of the Fed's emergency lending programs (such as the PDCF and the TSLF) during the 2007 -- 2009 Great Recession. \citet{bolton2011outside} also argued that the public liquidity programs, if implemented at a bad timing, will only crowd out the private liquidity supplied from the financial market.  As can be seen, our theory provides a different perspective on the same phenomenon.}


However, our result suggests that holdout by high-type firms 
is a blessing in disguise. The presence of these holdout firms is precisely what leads the market to make attractive offers to those that did not accept the bailout.  Indistinguishable from these holdout firms, those that actually sold to the market in $t=1$ also receive attractive offers, thus overcoming the adverse selection endemic to the problem. The alleviation of adverse selection for these firms in turn creates ``collateral benefits'' to those that accept the bailout, since no arbitrage means that they can also overcome stigma. Consequently, the increase in trade volume in $t=2$ more than makes up for the initial under-response, when compared with SSE.    

To an outsider observer, bailouts could very well be seen as ``working'' mysteriously here. Paradoxically, the policy invigorates the market by allowing high-type firms to send a strong signal on their financial strength by ``rejecting'' the bailout offer. This indirect signaling creates the desired stimulation effect later.  As Corollary \ref{cor1} suggests, the opportunity to ``reject" a bailout turns out to be more profitable than ``accepting" one. 


To establish the existence of DSE, it is necessary to check that buyers in each period do not have an incentive to deviate from the prescribed equilibrium strategies on every equilibrium path. Just as with SSE, it is difficult to check these ``no-deviation-by-buyers'' conditions.\footnote{Online Appendix \ref{Asec:DS} presents these necessary conditions and show them to be sufficient for equilibrium.}  Nevertheless, the equilibrium exists for a sufficiently large  $p_g$, as stated in Theorem \ref{thm:delay}.  For example, the equilibrium always  exists under the uniform distribution, as we show in  Section \ref{sec:uniform}. 
	
	
	Corollary \ref{cor1} also shows that, if a $p_g$ leads to a DSE under the transparent bailout, it yields the same total volume of trade as under a secret bailout. From this, one may conclude that  bailout stigma is not of economic concern in a DSE.  However, there are two caveats to this conclusion.  First, there is a possibility of multiple equilibria, as we show in the next section.  That is, the same $p_g$ may also support an SSE, which is clearly undesirable to a policy maker, as discussed previously.  In order to avoid the selection of such an equilibrium, a policy maker may have to raise $p_g$ beyond the lower bound identified in Theorem \ref{thm:nodelay}-(iii).  This, together with Theorems \ref{thm:delay}, suggests that a delayed equilibrium can be uniquely implemented by an offer $p_g\ge \max\{2\t_0 -\E[\t|\t\le \t_0], \E[\t| \t\in [\t_0, \gamma(\t_0)]]\}$.  Second, a delay in and of itself may be undesirable for reasons not modeled in our theory. For instance, prompt revitalization of economic activities often has external benefits for the rest of the economy.  In particular, if we consider financial institutions investing in the rest of the economy, a prompt restoration of their activities will have a positive spillover effect, and from this perspective delayed stimulation can be harmful. For these reasons, one may view the delay itself as a cost of stigma.	
	
 While it is difficult to find direct evidence on our equilibria, 
what transpired in the aftermath of the 2007--2009 Great Recession are somewhat consistent with our DSE. For instance, 
firms outright rejected rescue offers made by the government ostensibly to signal their financial strength, as mentioned in the introduction. In particular, our theory suggests that firms with high-quality assets are more likely to hold out when bailout terms are more generous (recall the difference between SSE and DSE.)  Indeed, there is some evidence that firms with high-quality assets participated less in programs that  were more generous. \citet{krishnamurthy2014sizing} find that, during the Great Recession, dealer banks with a large share of agency collaterals  (i.e., collaterals guaranteed by the US government) rarely participated in the Primary Dealer Credit Facility (PDCF) despite their favorable funding rates, although they did participate in the Term Securities Lending Facility (TSLF) along with banks with a large share of non-agency collaterals (i.e., collaterals not guaranteed by the US government).
	
	
\begin{remark}[The ``full'' freeze case: $\t_0 = 0$] \label{rem:full-freeze} \rm We have so far implicitly assumed  $\t_0>0$, for ease of exposition.  But our equilibrium characterizations from Theorems \ref{thm:nodelay} and \ref{thm:delay} extend even to the case of $\t_0 = 0$; that is, when the market would freeze completely absent any bailout.   In this case, our characterizations imply that $\hat \t=0$, and this leads to existence of a unique equilibrium.  Given $p_g \ge I$, the equilibrium admits a threshold  $\hat\t_g\in (0, p_g+S )$such that  types $\t \in [0, \hat\t_g]$ sell to the government in $t = 1$ but do not sell  in $t=2$, and types $\t \in (\hat\t_g, p_g+S ]$ hold out in $t=1$ but sell in $t=2$ at price 
$p_g$, where $\hat\t_g$ satisfies $\E[\t|\hat \t_g\le \t\le p_g+S]=p_g$.  Types $\t> p_g+S$ sell in neither period.
One can think of this  as a form of DSE: a bailout does not revive market in $t=1$ but induces  delayed trading in $t=2$.  In keeping with Corollary \ref{cor1}, the equilibrium induces the same trade volume as a secret bailout would. 
\end{remark}

\subsection{Uniform Distribution Example}	\label{sec:uniform}

To gain a better understanding about the two types of equilibria and their possible coexistence, it is useful to exhibit them in full detail for some concrete parameter values.  Assume uniform distribution $F(\t) = \t$, along with $S = 1/4$ and $I = 1/10$. We choose these parameter values for the sake of exposition, as shown in Figure \ref{fig:uni_ex} below. However, the qualitative results shown in Figure \ref{fig:uni_ex} hold more generally for $S < 1/2$ (so that $\t_0 = 2S < 1$) and $I < p_0 = S$.

\begin{figure}[htb]
	\centering
	\begin{subfigure}[b]{0.5\textwidth}
	\includegraphics[scale=0.6]{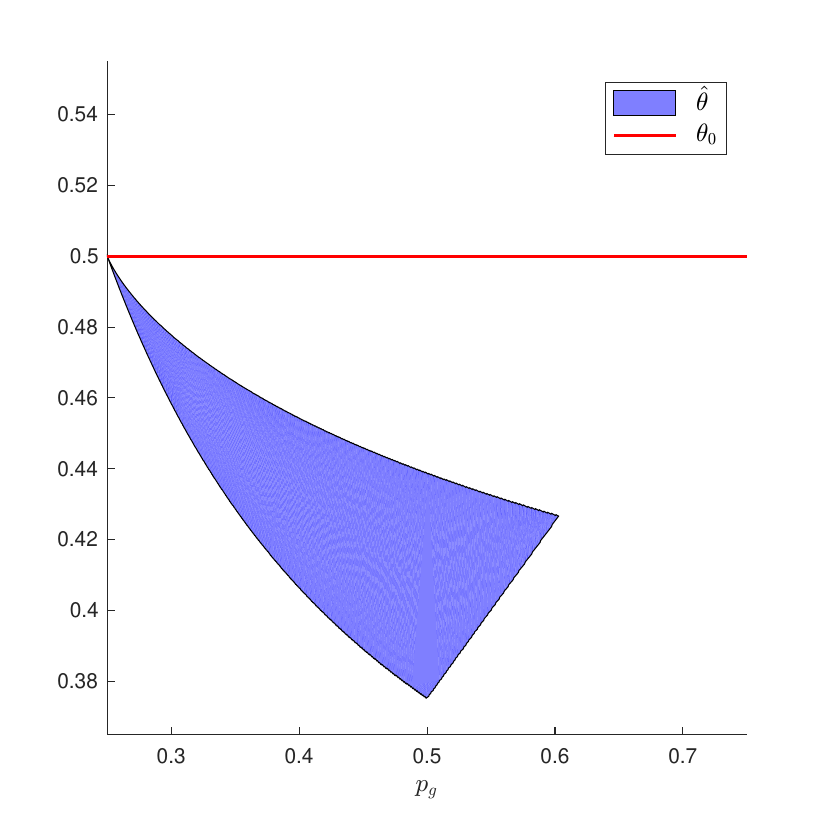}
	\caption{$\hat\t$'s in short-lived stimulation equilibria} \label{fig:uni_nodelay}
	\end{subfigure}
	\centering
	\begin{subfigure}[b]{0.48\textwidth}
	\includegraphics[scale=0.6]{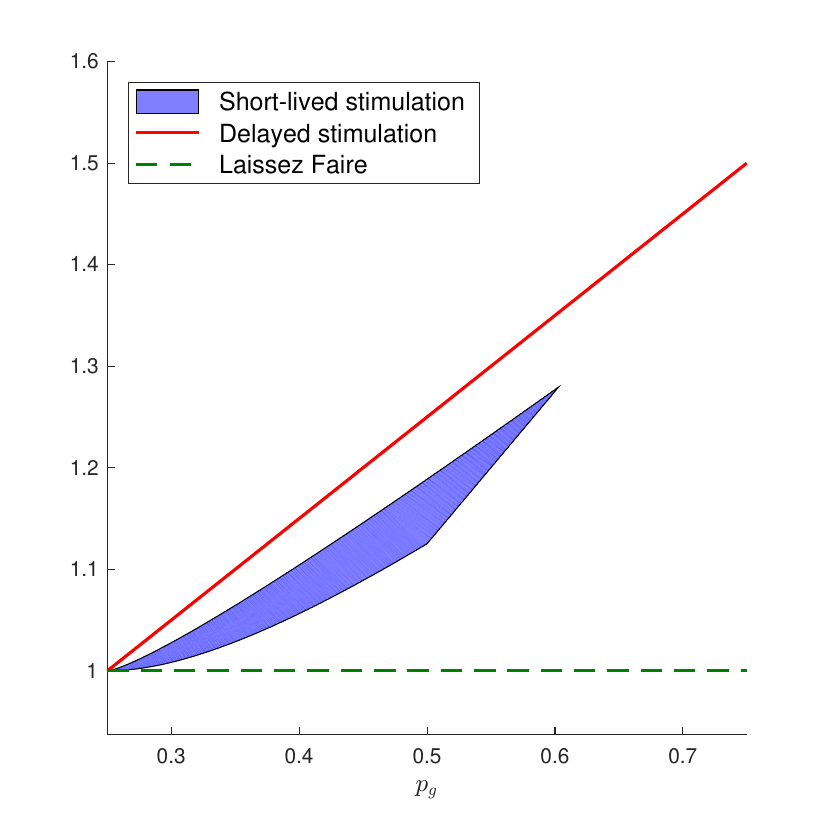}
	\caption{Overall trade in both types of equilibria} \label{fig:uni_trade}
	\end{subfigure}
	\caption{Equilibrium outcome in the uniform example ($S = 1/4$, $I = 1/10$)} \label{fig:uni_ex}
\end{figure}

The left panel (a) of Figure \ref{fig:uni_ex}  depicts the set of threshold types $\hat \t$  supported in SSE for various bailout terms $p_g$.  The corresponding threshold type for DSE is $\t_0$, as indicated by the straight red line. 
 The right panel (b) of Figure  \ref{fig:uni_ex} depicts the total volume of trade  induced by SSE (blue area) and DSE (red line) for differing levels of $p_g$.

This explicit example leads us to several interesting observations.  First, as can be seen clearly by the blue area,  {for a range of $p_g$'s, there is a continuum of SSE that induce different threshold values of $\hat \t$ for each $p_g$ in that range. } Note also that SSE exist for $p_g> p_0=S=1/4$, but do not exist when $p_g$ is sufficiently high.  Specifically, SSE exist only for $p_g<0.603$, which is well below the upper bound $2\t_0-\E[\t|\t\le \t_0]=0.75$ identified in Theorem \ref{thm:nodelay}-(iii).  Although not shown in the figure,  typically multiple DSE exist, but all of them induce the same threshold $\hat \t=\t_0$ and the same total trade volume, as formally stated in Theorem \ref{thm:delay}.

Second, both types of equilibria coexist for a range of bailout offers, as can be seen in the figure.  This multiplicity reflects the endogenous nature of equilibrium belief formation. Given a $p_g$, suppose a large measure of high-type firms accept a bailout but boycott the $t=2$ market to avoid stigma. This causes buyers to adjust down their offers in the $t=2$ market, in turn validating the firms' decision not to hold out in $t=1$. An SSE then arises.  By contrast, if a large measure of high-type firms hold out in $t=1$ for the same $p_g$, then  buyers in $t=2$ make high offers to attract them, which in turn validates their decision to hold out, leading to a DSE.

Third, as can be seen from  Figure \ref{fig:uni_ex}-(b), the effect of bailout can be discontinuous with respect to the terms of bailout.  Suppose the policy maker raises the bailout term $p_g$ starting from a low value close to $p_0=1/4$. At first, an SSE may arise.  As $p_g$ rises past $0.603$, however, SSE disappear. Instead, the market shifts to a DSE, characterized by an increase in trade volume at a delay. As mentioned before, one policy implication is that the policy maker may need to choose an attractive bailout offer ($p_g>0.603$ in this example) in order to avoid the selection of less desirable SSE.

Lastly, despite the stigma, a bailout with $p_g>p_0$ boosts overall trade regardless of the types of equilibria selected.  As shown in Figure \ref{fig:uni_ex}-(b), the total trade volume  in SSE is strictly higher than $2F(\t_0) = 1$, the total trade volume in the benchmark without bailout. This suggests that the positive effect of bailout on asset trading in $t = 1$---that is, every type $\t \le p_g + S$ sells in $t = 1$---outweighs the negative effect of stigma on asset trading in $t = 2$.   Further, as can be seen in Figure \ref{fig:uni_ex}-(b), the total trade volume in SSE tends to increase with $p_g$, so a more generous bailout tends to have a higher stimulation effect. 
	
	


\section{Cost of Bailouts and the Optimal Bailout Policy}	\label{sec:wel}

	So far, we have abstracted from the cost of bailouts, which is nevertheless a crucial element of policy debates.  In this section, we evaluate the effect of bailout programs while explicitly accounting for the cost of bailouts. To model the cost, we follow the literature (\cite{tirole2012overcoming}, \cite{chiu2016trading}) and assume that raising a dollar of public funds in use for the asset purchase program costs the society $(1+\l)$ dollars, where $\l\ge 0$ represents the deadweight loss of raising public funds, such as distortionary taxation. While  a budget deficit entails  social cost, we assume that a budget surplus does not contribute to any direct social welfare.  This is in keeping with the primary purpose of  government intervention here, which is to bail out distressed firms and to support socially desirable projects, and not to profit from the intervention itself.\footnote{\label{fn:no-profit}  This could be further justified by the assumption that profits from asset purchase are not easily converted to fund other government projects. If a budget surplus generates the commensurate welfare benefit---$\l$ per dollar---, the results below are still valid as long as $\l\ge\hat \t$.  }

	
	
	To evaluate the welfare of an equilibrium resulting from government bailouts, we adopt a mechanism design perspective.  In so doing, we allow for the policy maker to offer a menu of (arbitrarily numerous) bailout programs, not just a single purchase offer, and to adopt any information policy with regard to the identities of the firms that accept the programs.  Offering multiple bailout programs is realistic and policy relevant. During the Great Recession, numerous programs were introduced to help out financially distressed firms, which varied in bailout terms and generosity.  
	Similarly, information is an important tool for the policy maker; the policy maker may craft the disclosure policy to fine-tune the degree to which the identity of a bailout recipient is revealed to the market.  
	
	While being fully general in these two dimensions, in keeping with our baseline model, we maintain the following assumptions:  (a) the government programs are offered only in $t=1$; (b) the programs consist of a menu $\mathcal P_g$ of asset purchase offers  such that $p_g
	  \in [p_0, \overline P]$ for all $p_g\in \mathcal P_g$, for some arbitrarily high level $\overline P$, so that offered prices are at least as attractive as the {\it laissez-faire} price;\footnote{Recall that the assumption that government offers take the form of asset purchase programs is without loss, given that the investment returns are not pledgeable.  
	  } and (c) there is no rationing, meaning any (eligible) firms may choose any program made available by the government.  

	 Formally, we represent the equilibrium outcome of a  bailout policy by a pair of mappings, $(Q, T): [0, 1] \to \set{0, 1, 2}\times\mathbb{R}$,  where $Q(\t)\in \set{0,1,2}$ is a type-$\t$ firm's total asset sale across the two periods and $T(\t)$ is the total transfer it receives across the two periods in equilibrium.   {Naturally,  there exist $t=1$ and $t=2$ transfer functions $t_1(\t)$ and $t_2(\t)$ of $\t$, measurable with respect to $F(\cdot)$, such that $T(\t) = t_1(\t) + t_2(\t)$. } This transfer includes payment from both the government (if the firm accepts a bailout) and private buyers (if it sells to the market in either period).  
	  The fact that $Q$ is nonfractional reflects the no-rationing assumption mentioned earlier.   We can view these mappings as a \emph{direct mechanism} that implements a social choice as a function of report on the firm's type $\t$.  
	  
	  For the mechanism  $(Q, T)$ to represent an equilibrium outcome, it must satisfy incentive compatibility and participation constraints.  It is worth emphasizing that these constraints are required only for the total two-period allocation and transfers $(Q, T)$. In general, depending on the information  the policy maker makes available about the recipients of bailout programs in $t=1$,  further incentive constraints may be necessary for firms' $t=2$ payoffs. 
	  By ignoring these additional constraints, we focus only on the necessary incentive constraints that must hold regardless of the information policy. In this sense, we allow for a fully general information policy.\footnote{\label{fn:inscrutability} The preceding observation already implies that total secrecy,  which would obviate the need for additional incentive constraints for $t=2$, is at least weakly optimal. This follows from the standard insight in mechanism design:  secrecy, or ``no information,'' is tantamount to the pooling of incentive constraints across states, which relaxes the incentive constraints  by requiring them to hold only on average rather than state by state.} 
	  Given our incentive and participation constraints, one can invoke the celebrated \emph{revenue equivalence} or envelope theorem to characterize the welfare of a particular equilibrium via the trade volume and a payoff for a reference type (e.g., the highest type) induced by the equilibrium.  However, the additional constraints (a) and (b) make the standard method inapplicable.  The following characterization, which reflects (a) and (b), is crucial for our analysis. 
	
\begin{lemma} \label{lem:wel_fn}   Fix any equilibrium outcome  $(Q, T)$ in which a positive measure of firms accept some government bailout program.  Then,

\begin{description}
	
		\item [(i)] the total trade satisfies	\begin{align} \label{eq:general-Q}
	Q(\t) =  \begin{cases}
	2   & \mbox{ if } \t\le \hat\t,  \\  
	1
	&   \mbox{ if }  \t\in (  \hat\t , \t_2],\\
	0  & \mbox{ if }  \t> \t_2, \end{cases}
	\end{align} 
	for some  $\hat\t\le \t_0$ and $\t_2\ge \t_0$; and
	\item[(ii)] the social welfare is:
	\begin{align} \label{eq:wel}
	\int^1_0 \left\{2  \t+ \left(S(1+\l)  - \l \frac{F(\t)}{f(\t)}\right)Q(\t) -\lambda (\overline u-2)\right\}  f(\t) d\t,
	\end{align}
	for some payoff $\overline{u} \ge 2$ accruing to the highest-type ($\t = 1$). 
\end{description}	
\end{lemma}
\begin{proof} See Appendix \ref{app:lem2}.
 	\end{proof}
Part (i) reflects the standard monotonicity condition necessary for incentive compatibility, in the presence of single crossing, as was the case with Lemma \ref{lem:char}.  The fact that the equilibrium trade volume takes a particularly simple form follows from (c), that is, the lack of rationing.  The  crucial part is that $\hat\t\le \t_0$---or $Q(\t)\le 1$ for all $\t\ge \t_0$. This property, which reveals a fundamental limit to the effect of trade stimulation, follows from (a). Recall from Theorems \ref{thm:nodelay} and \ref{thm:delay} that both  types of equilibria in {Section} \ref{sec:gov} also exhibit the same limit; the claim here is that the limit applies even when the government offers  a menu of bailout programs.  Lastly, the fact that $\t_2\ge \t_0$, which follows from (b),  simply means that the government will never gain from suppressing the trade below the {\it laissez-faire} level.  

Part (ii) shows that  the welfare under a general bailout equilibrium consists of three terms.  The first term, $2\t$, is simply the value of a type-$\t$ asset; recall that each firm owns 2 units of such an asset, and its value does not depend on the eventual owner of that asset. The second term $\left(S(1+\l)  - \l \frac{F(\t)}{f(\t)}\right)Q(\t)$ describes the social welfare associated with inducing a trade $Q(\t)$ for a type-$\t$ firm.  Sale of an asset by a type $\t$-firm yields social benefit of $S$, the return to the project enabled by the sale.  Meanwhile, it creates a budget shortfall of $\frac{F(\t)}{f(\t)}-S$, which inflicts a cost of  $\l(\frac{F(\t)}{f(\t)}-S)$ to society.\footnote{\label{fn:shortfall} To check the marginal budget shortfall $ \frac{F(\t)}{f(\t)}-S$ caused by a sale by the type-$\t$ firm, note that $$ \int_{0}^1(T(\t)-\t Q(\t))f(\t)d\t=\overline{u}-2+ \int_{0}^1 \left(\frac{F(\t)}{f(\t)}-S \right)Q(\t)f(\t)d\t,$$
	where we used the envelope theorem to substitute for $T(\t)$.   Intuitively, the inverse hazard term,  $\frac{F(\t)}{f(\t)} Q(\t)$,  is the incentive cost that is required to induce type-$\t$ firm to sell an additional unit of its asset:  since any sale by type $\t$ can be mimicked by all lower types,  these types must be paid rents to prevent their mimicking.  The project return  $S$ acts as an incentive for sale, and thus mitigates the budget shortfall and saves the subsidy needed to induce a sale.}
Finally, a rent of $\overline{u}-2$ accrues to the highest-type ($\t=1$) firm (namely, above its asset value 
$2$) and adds to the budget shortfall, inflicting a social cost of $\l (\overline{u}-2)$.\footnote{{By the standard envelope theorem reasoning, any rent enjoyed by the highest type is enjoyed by  ``all'' firm types. Effectively, such a rent amounts to a pure cash transfer.  Even though we do not formally include pure subsidy as a possible policy tool, our model implicitly allows for such a possibility.  Evidently, pure subsidy is never optimal here as it is dominated by asset purchase which entails a lower deficit.}}  

If no firm accepts a government program in equilibrium, then the outcome is the same as the no-bailout benchmark described in Section \ref{sec:benchmark}. The resulting trade rule $Q$ is still  described by (\ref{eq:general-Q}) with $\hat{\t}=\t_2=\t_0$, and the associated welfare is accurately characterized by (\ref{eq:wel}).\footnote{In that case, the budget shortfall reduces to zero, so the terms multiplied by $\l$ vanish.  See footnote \ref{fn:shortfall}.}



Inspecting the welfare objective in (\ref{eq:wel}), it is immediate that the rents accruing to the highest type serves no purpose; they only entail social cost without stimulating any trade. Hence, the optimal policy has $\overline u=2$; this is accomplished by limiting the generosity of offers.  In light of Lemma \ref{lem:wel_fn}, the mechanism design problem is therefore reduced to finding two cutoffs $\hat \t$ and $\t_2$ that solve: 

$$\max \,  2\int_0^{\hat{\t}} \left(S(1+\l)  - \l \frac{F(\t)}{f(\t)}\right)f(\t)d\t + \int_{\hat{\t}}^{\t_2} \left(S(1+\l)  - \l \frac{F(\t)}{f(\t)}\right)f(\t)d\t$$
subject to:
$$\hat\t\le \max\{ \t_0, \t_2\}.$$
The optimal solution depends on the cost $\l$ of public fund, in particular, whether it is less than a threshold   
$\hat \l:=\inf \left\{\tilde\l\ge 0: S(1+\tilde\l)\ge \tilde\l\frac{F(\t_0)}{f(\t_0)}\right \}$.
	
\begin{theorem}\label{thm:wel_com}  The following characterizes the optimal bailout policy.
 
\begin{description}
	\item[(i)] The uniquely optimal outcome is characterized by a trade rule $Q$ in (\ref{eq:general-Q}) with cutoffs $\hat{\t}=\t_0$ and $\t_2=\max\{\t_0, \t_2^*\}$, where $ \t^*_2 :=\sup\{\t\in [0,1]:(1+\l)S\ge\l \frac{F(\t)}{f(\t)}\}$, with no rents accruing to the highest type.\footnote{Here, the ``unique'' optimality means that any equilibrium outcome $(Q,T)$ where $Q$ differs from $Q^*$ for a positive measure of $\t$'s is strictly welfare dominated by $(Q^*, T^*)$, where $T^*$ is the transfer rule implementing $Q^*$ with zero rents accruing to the highest type.}
	\item [(ii)]  If $\l\ge\hat{\l}$, then the optimal outcome is implemented by no intervention; and if $\l<\hat{\l}$, then it is implemented by a secret bailout or, equivalently, a delayed-stimulation equilibrium under a single purchase offer of $p_g:= \t_2^*-S.$
	
	\item [(iii)] The short-lived stimulation equilibrium is strictly suboptimal.
	
	\item [(iv)] A menu of multiple transparent bailout offers is strictly suboptimal as long as positive measures of firms accept distinct bailout offers in equilibrium.  
\end{description}

\end{theorem}

\begin{proof} See Appendix \ref{app:thm5}. \end{proof}

 Theorem \ref{thm:wel_com} summarizes several broad lessons learned from our general mechanism design analysis.  Parts (i) and (ii) of the theorem suggest that the best possible outcome that the government can achieve with its bailout policy can be achieved by a \emph{single} bailout offer made in total {\it secrecy}.
 A transparent bailout that leads to the delayed-stimulation equilibrium can also achieve the same outcome, but the possible multiplicity of equilibria complicates the problem, as we discuss below.

 First, the optimality of total secrecy means that even if the government can design its information policy with regard to firms' acceptance of the offers without any constraints, it can never do strictly better than keeping the market in the dark about the identities of the bailout recipients. This follows from the fact that the information revelation and separation of types exacerbate adverse selection and market unraveling (recall footnote \ref{fn:inscrutability}).  In light of  Corollary \ref{cor1}, this means that the optimum can be attained by the DSE under transparency, which is important in practice given the difficulty of maintaining secrecy.  In particular, this implies that the SSE is strictly worse, as stated in part (iii).  Note that this conclusion is not {\it a priori} obvious from Corollary \ref{cor1}, since for any DSE associated with $p_g$, an SSE with some suitably chosen $p_g'$ may implement the same total volume of trade.  What part (iii) suggests is that even when such an offer exists, it would require a strictly higher deficit on the part of the government and therefore a strictly  higher social cost. Intuitively, a more severe stigma arising from that equilibrium requires a larger deficit to support the same level of stimulation.  Admittedly, the bailout offer $p_g$ that implements the optimal outcome under a DSE may also admit a less desirable SSE.  To avoid this situation,  the policy maker may need to raise the offer even further to a level for which the latter type of equilibrium is no longer possible, although this may not implement the optimal policy. {This result is also consistent with the Fed's decision to lower the interest rate charged to the banks borrowing from the Discount Window in the early phase of the Great Recession.}  

 The second implication of part (ii) is that it does not pay the policymakers to offer multiple bailout programs. Indeed, part (iv) shows that doing so under full transparency is strictly suboptimal in {\it any} equilibrium in which firms self-select into multiple programs.  This finding does not necessarily contradict  the approach taken by the US government in the wake of {\it Great Recession}.  The plethora of bailout programs rolled out by the government targeted different segments of the financial market and the broader economy, whereas part (iv) concerns the multiplicity of programs targeting the {\it same} segment of the economy. {One example of the latter case is the coexistence of the Discount Window (DW) and the Term Auction Facility (TAF), both of which targeted  depository institutions. Our result suggests that the Fed could have improved the efficacy of the public liquidity provision if it had chosen either DW (at a very low borrowing rate) or TAF alone, but not both. Running both programs could have worsened the overall perception of not only the banks using the DW facility, but also the TAF users and the banks that did not participate in either program. Put differently,} it is never desirable to offer a menu of multiple programs solely for the purpose of eliciting firms' types. Doing so backfires by compounding the incentive constraints. The multiplicity of programs leads to more separation of firm types, which exacerbates adverse selection and market freeze.

\section{Related Literature}	\label{sec:lit}

While the broad theme of this paper is related to an extensive literature on the benefits and costs of government intervention in distressed banks,\footnote{The primary rationale for intervention is to prevent the contagion of bank runs whether it stems from depositor panic (\citet{diamond1983bank}), contractual linkages in bank lending (\citet{allen2000financial}), or aggregate liquidity shortages (\citet{diamond2005liquidity}, \citet{diamond-rajan2012}). The costs of anticipated bailouts due to the time-inconsistency of the policy  are discussed by, among others, \citet{stern2004too}; the optimal response to the time inconsistency problem is studied in \citet{green}, \citet{chari2016}, and \citet{keister}.}  our work is most closely related to \citet{philippon2012optimal} and \citet{tirole2012overcoming}, who focus on  adverse selection in asset markets as a primary reason for government intervention.\footnote{Regarding the optimal form of bailouts, \citet{philippon2012optimal} show that optimal interventions involve the use of debt instruments when adverse selection is the main issue. With additional moral hazard but limits on pledgeable income, \citet{tirole2012overcoming} justifies asset purchases. When there is debt overhang due to lack of capital, \citet{philippon2013efficient} find that optimal interventions take the form of capital injection in exchange for preferred stock and warrants. During the US subprime crisis, the Emergency Economic Stabilization Act (EESA) initially granted the Secretary of the Treasury authority to purchase or insure troubled assets owned by financial institutions.  However, the Capital Purchase Program under the TARP switched to capital injection against preferred stock and warrants.}  As mentioned previously, these studies do not explicitly study the dynamic consequence of receiving a bailout---the focus of the current study. Even though these papers recognize that relatively low types accept bailouts, this does not translate into an adverse effect  on subsequent financing in their models. Our dynamic model captures not only how bailout stigma affects firms' financing behavior but also how the stigma fundamentally alters the role of a bailout.  In particular, our most striking takeaway is the ability of firms to favorably signal their type by ``refusing'' an attractive bailout offer, which has no analogue in these or other antecedent studies.


Banks' reputational concerns   are explicitly considered in \citet{ennis2013over}, \citet{jennifer2014predatory}, \citet{chari2014reputation}, \citet{ennis2019interventions},  and \citet{hu-zhang}. In \citet{ennis2013over}, to meet  their short-term liquidity needs, banks with high-quality assets use interbank lending while those with low-quality assets use the discount window. The resulting discount window stigma is reflected in the subsequent pricing of assets. In \citet{jennifer2014predatory}, financially strong banks use the Federal Reserve's Term Auction Facility  since winning the auction at a premium signals financial strength, which protects them from predatory trading.  {\citet{hu-zhang} study the banks' choice between the Discount Window and Term Auction Facility (TAF) where funds are immediately available from the Discount Window, albeit entailing a stigma, whereas TAF releases funds with delay. They show how banks self-select themselves into different programs.} The main focus in \citet{chari2014reputation} is on how reputational concerns in  secondary loan markets can result in persistent adverse selection. Since all three studies consider discrete types of banks and there is no government bailout, their results are not directly comparable to ours. \citet{ennis2019interventions} extends \citet{philippon2012optimal} by allowing banks to borrow from the discount window before borrowing from the market, thereby formalizing the discount window stigma alluded to in \citet{philippon2012optimal}. There are two main differences between \citet{ennis2019interventions} and our work.  First, there is only one investment opportunity in \citet{ennis2019interventions}, hence it does not capture the dynamics of a  stimulation effect (e.g., the possibility of stimulation being either short-lived or delayed).  Second,  \citet{ennis2019interventions} characterizes equilibria for exogenously given discount window rates, whereas we study the optimal bailout policy. 


Our paper is also related to  studies on dynamic adverse selection in general (\citet{inderst2002competitive}, \citet{janssen2002dynamic}, \citet{moreno2010decentralized}, \citet{camargo2014trading}, \citet{fuchs2015government}) and those with a specific focus on the role of information in particular (\citet{horner2009public}, \citet{daley2012waiting}, \citet{fuchs2016transparency}, \citet{kim2017information}).\footnote{Others include dynamic extensions of Spence's signaling model with public offers (\citet{noldeke1990signalling}), private offers (\citet{swinkels1999education}), and private offers with additional public information such as grades (\citet{kremer2007dynamic}).}  The key insight from the first set of studies is that dynamic trading generates sorting opportunities, which are not available in the static market setting. However, each seller has only one opportunity to trade in these studies, so signaling is not an issue. The second set of studies relates to different disclosure rules and how they affect dynamic trading. For example, \citet{horner2009public} and \citet{fuchs2016transparency} show that secrecy (private offers) tends to alleviate adverse selection but transparency (public offers) does not. Once again,  each seller has only one trading opportunity in these studies.  Hence, although past rejections can boost reputation, acceptance ends the game. In contrast, in our model, acceptance as well as rejection of the bailout offer work as signaling opportunities.  Although our model also shows that secret bailouts weakly dominate transparent bailouts, none of these papers studies government intervention in response to market failure.

There are several empirical studies that provide evidence on stigma in the financial market.  \citet{peristiani1998growing} provides early evidence on the discount window stigma. \citet{furfine2001reluctance, furfine2003standing} finds similar evidence from the Federal Reserve's Special Lending Facility during the period 1999-2000 and the new discount window facility introduced in 2003.  
\citet{armantier2015discount} utilize the Federal Reserve's Term Auction Facility bid data from the 2007-2008 financial crisis to estimate the cost of stigma and its effect.  \citet{cassola20132007} find evidence of stigma from the bidding data from the European Central Bank's auctions of one-week loans. \citet{krishnamurthy2014sizing} find  that in repo financing, dealer banks with higher shares of agency collateral  repayments, that is, collateral (implicitly) guaranteed by the government, borrowed less from the Primary Dealer Credit Facility (PDCF) despite its attractive  funding terms, which indeed supports the existence of a stigma attached to the users of the PDCF. 

Finally, there is a vast literature discussing the 2008 Great Recession. \citet{congressional2012} estimates the overall cost of the TARP at approximately \$32 billion, the largest part of which stems from assistance to AIG and the automotive industry while capital injections to financial institutions are estimated to have yielded a net gain. For detailed assessments of the various programs in the TARP, see the {\it{\citet{JEP2015}}}. See also \citet{fleming2012federal}, who discusses  how the various emergency liquidity facilities provided by the Federal Reserve during the 2007-2009 crisis were designed to overcome the limitations of traditional policy instruments at the time of crisis. \citet{tong2020did} provide international evidence on the effect of unconventional interventions during 2008-2010.

\section{Conclusion}		\label{sec:conclusion}

The current paper has studied a dynamic model of a government bailout in which firms have a continuing need to fund their projects by selling their assets. Asymmetric information about the quality of assets gives rise to adverse selection and a concommitant market freeze, which provides a rationale for a government bailout, just as in \citet{tirole2012overcoming}. However, in contrast to  \citet{tirole2012overcoming}, markets stigmatize bailout recipients, which jeopardizes their ability to fund subsequent projects. The presence of this bailout stigma and other dynamic incentives yields a much more complex and nuanced portrayal of how bailouts impact the economy than have been recognized in the extant literature.

Our main findings can be summarized as follows. Bailout stigma necessitates the government to pay a premium over the market terms to compensate for the stigma. 
Even so, market rejuvenation can be short-lived and adverse selection can worsen in subsequent market trading, resulting in a market freeze even more severe than in the absence of a bailout.  This requires the government to further increase a premium.  A more attractive bailout premium can be effective in stimulating trade and investment, but its effects are delayed.  Delayed benefits materialize as bailouts provide firms with opportunities to boost reputation by ``rejecting'' bailout offers.  This improves their ability to trade in the market in  subsequent periods. Indeed, there is no welfare loss in this case relative to a secret bailout, {which attains   optimal welfare.} As such, the delayed effects of bailouts can be a blessing in disguise, subject to two important caveats: the government may need to run a large budget deficit to support delayed market stimulation, and delay in and of itself may be undesirable for reasons not modeled in the current paper.

The central lesson from the current work is that, compared with the static setting, the dynamic effects of bailouts are very different due to interplay between the bailout stigma, the market's belief within and across periods, and rich signaling opportunities firms have in the dynamic context. To the best of our knowledge, the insights we develop and the forces we identify are novel and have not been recognized in the previous literature, and should be part of the  framework for conducting future policy debates and empirical studies.


\bibliographystyle{aer}
\bibliography{stigma}

\appendix
\renewcommand{\thesection}{A.\arabic{section}}
\renewcommand{\theequation}{A\arabic{equation}}

\section*{Appendix: Proofs}


\section{Proof of Theorem \ref{thm:tirole}} \label{app:thm1}

We have already established (i) in the paragraphs forth to Theorem \ref{thm:tirole}. To prove (ii), suppose to the contrary  $\bar\t_m\le \bar\t_g$.   Then, by (\ref{eq:feas}), $\bar\t_m\le  \E[\t|\t\le p_g+S]$.  Since $\bar\t_m=p_g$, we have $p_g\le  \E[\t|\t\le p_g+S]$, or $p_g+S\le  \E[\t|\t\le p_g+S]+S$. By the definition of $\t_0$, 	$p_g+S\le \t_0 = p_0+S$, which contradicts $p_g>p_0$.  The remaining characterizations follow from the observations preceding the theorem.  	  $\square$

\section{ Proof of Theorem \ref{thm:nodelay}} \label{app:thm3}
Fix an SSE. Before proceeding with the proof, we first make several  preliminary observations. 

To begin, consider types $\t \le \hat\t$.  Let  $\mu_g$ and $\mu_m$ be the fractions of these firms selling to the government and to the market in $t=1$,  respectively, where $\mu_g+\mu_m=1$.
Let $\bar\t_g$ and $\bar\t_m$ denote respectively the mean values of the assets sold by the two groups. By definition, 
\begin{align}
\label{eq:feas2}
\mu_g \bar\t_g+\mu_m\bar\t_m & =\E[\t|\t\le \hat\t].
\end{align}
 Let $p_m$ denote the price firms receive from selling to the market in $t=1$.  The $t=2$ market price depends on whether or not a firm received the bailout in $t=1$, as these events are observed by buyers in $t=2$. Let $p_2^g$ and  $p_2^m$ denote respectively the prices offered in $t=2$ to those firms that accepted the bailout and to those firms that did not in $t=1$.  The latter group consists of only those that sold to the market in $t=1$ because, in SSE,  there do not exist firms selling only in $t=2$.  Since buyers break even in expectation, we must have $p^g_2=\bar\t_g$.
Further,  $p_2^g=\bar\t_g\ge I$, or else these firms would not sell in $t = 2$, a contradiction to SSE. Similarly, $p_m=p^m_2= \bar\t_m$, since those that sold to the market in $t=1$ are also believed correctly to be of type $\bar\t_m$ on average in both periods. 

 {We first prove that $\mu_g>0$ and $\mu_m>0$. To see this, suppose first $\mu_m = 0$.  Then, all types $\t\le \hat\t_g$ sell to the government in $t=1$.  But then, the holdouts in $t=1$ would be revealed in $t=2$ to have types $\t>\hat\t_g$.  Given $\hat\t_g<1$, a positive measure of them will attract buyers offering price higher than $\hat\t_g$. This leads to $\t_2>\hat\t_g$, a contradiction to the type of equilibrium we are considering.  
	Suppose next $\mu_g=0$ so that all types $\t\le \hat{\t}$ sell to the market.  We cannot have $\hat{\t}=\hat{\t}_g$, since then no firm accepts bailout, so the {\it laissez-faire} outcome would prevail with $\hat{\t}=\hat{\t}_g=\t_0=p_0+S$.  But then some types $\t\in [\t_0,p_g+S]$ would deviate and sell to the government at $p_g>p_0$ (and say withdraw from the $t=2$ market).  Hence, $\hat{\t}<\hat{\t}_g$.  But then, bailout recipients are revealed to have types $\t\ge \hat{\t}$, so they attract offers higher than $\hat\t$ in $t=2$.  This contradicts the definition of $\hat{\t}$ in Lemma \ref{lem:char}.}

Further, those firms selling in both periods must be indifferent between accepting the bailout and selling to the market in $t = 1$: 
\begin{align}
\label{eq:arb}
p_g + p^g_2 + 2S &= p_m+p_2^m+2S\, \Longleftrightarrow  \,
p_g+ \bar\t_g=2 \bar \t_m.
\end{align}

Next, suppose further that $\hat\t<\hat\t_g$.  Then,  the cutoff type $\hat{\t}$ must be indifferent between selling to the market in both periods and accepting the bailout in $t=1$ and not selling in $t = 2$:
\begin{align} 
\label{eq:indiff}
2 \bar\t_m +2S = \hat\t+ p_g + S.
\end{align}

Lastly,  the type  $\hat\t_g$ must be indifferent between accepting the bailout in $t = 1$ and not selling in $t = 2$ \emph{and} not selling in either period. 
Hence,  
\begin{align} \label{eq:holdout}
\hat\t_g = p_g + S,
\end{align}
which by assumption is less than one. 

These necessary conditions  (\ref{eq:feas2})--(\ref{eq:holdout}) restrict the equilibrium variables $(\mu_g, \mu_m, \hat \t, \hat\t_g, \bar\t_g, \bar\t_m)$ given $p_g$.  We are now ready to prove the statements of Theorem \ref{thm:nodelay}. 


\paragraph{\bf Proof of Part (i):  $\bar\t_g < \bar\t_m < p_g$.}   We first establish the following claim. 

\smallskip 
\begin{claim} \label{claim1}
$\bar\t_g<\bar\t_m$.  
\end{claim}

{\it Proof:}  Suppose to the contrary that $\bar\t_g\ge\bar\t_m$, which in turn implies $p_g \le \bar\t_m$ from \eqref{eq:arb}. Given Lemma \ref{lem:char}, there are two possible cases to consider: $\hat{\t} < \hat{\t}_g$ or $\hat{\t} = \hat{\t}_g$. 

Suppose first $\hat{\t} < \hat{\t}_g$.  Then, we have $\hat\t_g \le p_g + S$ from \eqref{eq:holdout}. Moreover, we have from (\ref{eq:feas2}) that $\bar\t_m\le \E[\t|\t\le\hat \t]\le \E[\t|\t\le p_g+S]$. Since $p_g>p_0$, we must have $\E[\t|\t\le p_g+S]<p_g$, or else $p_g\le p_0$ (recall definition of $p_0$ from \eqref{eq:theta0} as well as Figure \ref{fig:akerlof}). Hence, $\bar\t_m<p_g$, which, however, contradicts the earlier hypothesis $p_g \le \bar\t_m$. 

Suppose next  $\hat\t = \hat\t_g$. Then, a type-$\hat\t$ firm must be indifferent between selling in both periods and selling in neither period if $\hat\t < 1$ and (weakly) prefer selling in both periods if $\hat\t = 1$. Hence, we must have $2 \hat\t \le 2 \bar\t_m + 2S$, where the inequality  holds strictly only if $\hat\t = 1$.  We thus conclude that $\hat\t = (\bar\t_m + S)\wedge 1$. Moreover, we have from \eqref{eq:feas2} that $\bar\t_m \le \E[ \t | \t \le \hat\t] \le \E[ \t | \t \le \bar\t_m + S]$, which implies $\bar\t_m \le p_0$ (again by the definition of $\t_0$ in  \eqref{eq:theta0}).   Since  $p_g>p_0$, this again contradicts an earlier hypothesis that $p_g\le \bar \t_m$.    We thus conclude that $\bar\t_g<\bar\t_m$.  $\square$.

\vskip 0.3cm

 By (\ref{eq:arb}), Claim 1 in turn implies that  $\bar\t_m < p_g$.  We have thus proven (i).   $\square$
 
 \vskip 0.3cm  
 
 \noindent  {\bf Proof of  Part (ii):   $\hat{\t} < \t_0 < \hat\t_g = p_g + S$.}   We first establish the following two claims: 
 \smallskip 
 

\begin{claim} \label{claim2}
 $\hat\t \neq 1$.  
\end{claim}

 {\it Proof:} 
 Suppose to the contrary that $\hat\t = 1$. Then,  $\hat\t = 1 \le \bar\t_g + S$.  Otherwise, we will have $p_g + S + \hat\t > p_g + \bar\t_g + 2S = 2\bar\t_m + 2S$, where the equality is from  (\ref{eq:arb}), so   a type-$\hat\t$ firm would deviate by accepting the bailout but boycotting the $t = 2$ market.  Since $\bar\t_g < \bar\t_m$ (by Claim \ref{claim1}) and $\mu_g \bar\t_g + \mu_m \bar\t_m = \E[ \t | \t \le \hat\t] = \E[\t | \t \le 1]$, we have
 $$1\le \bar\t_g + S<  \mu_g  \bar\t_g + \mu_m \bar\t_m +S \le \E[\t | \t \le 1] + S.$$ But this contradicts $\t_0 < 1$, which we assume throughout.   $\square$
 

\begin{claim} \label{claim3}
  $\hat\t < \hat\t_g$.  
\end{claim}
 
 {\it Proof:} 
 Suppose to the contrary that $\hat\t = \hat\t_g$.  By Claim \ref{claim2}, we  have $\hat\t = (\bar\t_m + S) \wedge 1$ and $p_g > \bar\t_m$. In equilibrium, each type $\t > \hat\t$ never sells in either period and obtains payoff $2 \t$. Since $p_g > \bar\t_m$, however, types $\t \in (\hat\t_g, p_g + S)$ will have a strictly higher payoff than $2\t$ by selling to the government in $t = 1$, a contradiction.  $\square$
 
 \medskip
 We are now ready to prove (ii). We first show that $\hat\t_g > \t_0$. Since $p_g > p_0$, it is straightforward from \eqref{eq:holdout} that $\hat\t_g = (p_g + S) \wedge 1 > \t_0$. We next prove $\hat{\t}<\t_0$. Since $\hat\t < \hat\t_g$ by Claim 3, we have
	$$\hat{\t}=\bar\t_g+S< \E[\t|\t\le \hat{\t}]+S,$$
	where the equality follows from (\ref{eq:arb}) and (\ref{eq:indiff}), and the strict inequality follows from $\bar\t_g<\bar\t_m$ and (\ref{eq:feas2}).  The definition of $\t_0$ then implies $\hat{\t}<\t_0$.    $\square$
	
\medskip 
 \noindent 	
 {\bf Proof of  Part (iii):} {\it an SSE exists only if $p_g < 2\t_0 - \E[ \t | \t \le \t_0]$.}  To prove this, observe from \eqref{eq:arb} that
	$$p_g = 2\bar\t_m - \bar\t_g.$$
	Fixing $\hat\t$, the RHS is maximized  when, for some threshold $\tilde \t\in [0, \hat{\t}]$, all types $\t>\tilde \t$ sell to the market and all types $\t<\tilde \t$ sell to the government so that  $\bar\t_m = \E[ \t | \t \in (\tilde\t, \hat\t]]$ and $\bar\t_g = \E[ \t | \t \le \tilde\t]$.  Hence, 
	\begin{align*}
	p_g  &= 2\bar\t_m - \bar\t_g \\
		&\le \max_{\tilde\t \in [0, \hat\t], \hat\t \in [0, \t_0)} 2\E[ \t | \t \in (\tilde\t, \hat\t]] - \E[ \t | \t \le \tilde\t] 	\\
		&< \max_{\tilde\t \in [0, \t_0]}  2\E[ \t | \t \in (\tilde\t, \t_0]] - \E[ \t | \t \le \tilde\t]  \\
		&= 2\t_0 - \E[\t | \t \le \t_0],
	\end{align*}
	where the strict inequality follows from $\hat\t < \t_0$ and $\E[ \t | \t \in (a, b)]$ is increasing in $b$ for all $0 \le a \le b \le 1$, and the last equality follows from the regularity condition that $2\E[ \t | \t \in (a, b)] - \E[ \t | \t \le a]$ is increasing in $a$ for all $0 \le a \le b \le 1$. $\square$

\section{Proof of Theorem \ref{thm:delay}}  \label{app:thm4}

The existence result and the required condition for $p_g$' are established in  Proposition \ref{Aprop:DS_high} in Online Appendix.  Here, we focus on the characterization of equilibria. As before, let $\mu_g$ and $\mu_m$ respectively denote the fractions of types $\t\le\hat\t$ that sell to the government and to the market in $t=1$, and let $\bar\t_g$ and $\bar\t_m$ denote their average values. Obviously, (\ref{eq:feas2}) must continue to hold. Let $p_m$ be the market price for the asset in $t=1$, and let $p_2^g$ and $p_2^m$ respectively denote the $t = 2$ prices for those that sold to the government and those that did not. Note that $p_2^m$ applies to those that sold to the market in $t=1$ and to those that held out, since $t=2$ cannot distinguish them.

We wish to prove that $\hat\t=\t_0$, $\t_2= p_g+S$, and $\mu_g > 0$. There are two possible cases: $\hat{\t}_g>\hat{\t}$ and $\hat{\t}_g = \hat{\t}$, and we treat them separately. (Recall by definition  $\hat{\t}_g\ge\hat{\t}$.)

\subsection{The case of $\hat{\t}_g>\hat{\t}$.}    

	In this case, firms with $\t \in (\hat\t, \hat\t_g]$ sell to the government in $t = 1$. Obviously, $\mu_g \ge \int^{\hat\t_g}_{\hat\t} f(\t) d\t > 0$. We only need to prove $\hat\t = \t_0$ and $\t_2 = p_g + S$. 
 	
	We first show $\t_2 = p_g + S$. Observe that a type-$\hat{\t}_g$ firm must be indifferent between selling only in $t=1$ to the government at $p_g$ and selling only in $t=2$ at price $p_2^m$.  Hence, we must have $p_2^m=p_g$. Since a type-$\t_2$ firm must be indifferent between selling only in $t = 2$ at price $p^m_2$ and not selling in any period, we must have $\t_2 = p^m_2 + S = p_g + S$.

	We next show $\hat\t = \t_0$. We restrict our focus on the case $\mu_g>0$ and $\mu_m>0$ (The argument for the other case $\mu_m = 0$ is similar, so we omit the proof). Then, these firms sell in both periods, and thus must be indifferent between selling to the government and to the market in $t=1$.  This implies (\ref{eq:arb}), or
$$p_g+ p_2^g=p_m+p_2^m=p_m+p_g \Rightarrow p_m = p_2^g.$$ 
This, together with the zero-profit condition, implies that
\begin{equation}\label{eq:delay-arb}
\bar\t_g=p_2^g=p_m=\bar\t_m.
\end{equation} 
It then  follows from (\ref{eq:feas2}) that 
\begin{equation}\label{eq:delay-feas}
\bar\t_g =\bar\t_m=\E[\t|\t\le\hat\t].
\end{equation}
	Next, a type-$\hat{\t}$ firm  must be indifferent between selling in both periods and selling only in $t=1$ to the government at price $p_g$:
\begin{equation}\label{eq:delay-indiff}
p_m+p_g+2S= p_g+S+\hat\t \, \Longleftrightarrow \, p_m+S=\hat \t,
\end{equation} 
 which, together with (\ref{eq:delay-arb}) and (\ref{eq:delay-feas}), implies that
 $$ \E[\t|\t\le\hat\t]+S=\hat \t.$$
By definition of $\t_0$, or \eqref{eq:theta0}, we then have $\hat{\t}=\t_0$.  This in turn implies $p_m=p_2^g= \E[\t|\t\le \t_0]=p_0$. 
 
\subsection{\bf The case of $\hat{\t}_g=\hat{\t}$.}

The proof proceeds in several claims.  

  \begin{claim} \label{claim4}
  	$\mu_g > 0$.   
  \end{claim}
  
 {\it Proof:}    Suppose to the contrary that  $\mu_g = 0$. Then, we must have $p^m_2 = p_0$ since buyers in $t = 2$ do not observe any action taken by firms in $t = 1$. This implies $p_0 = p^m_2 = p_g > p_0$, a contradiction.  $\square$

\begin{claim} \label{claim5}
	  $\t_2 = p_g + S$ { and $p_2^m=p_g$.}
\end{claim}

 {\it Proof:}   Since types $\t\in(\hat{\t}, \t_2)$ must weakly prefer selling only in $t=2$ at price $p_2^m$ to selling only in $t=1$ to the government at price $p_g$. This implies $p_2^m\ge p_g$.  We now prove that $p_2^m=p_g$.  Suppose to the contrary that $p_2^m>p_g$. 

We know from Claim \ref{claim4} that $\mu_g > 0$. Suppose $\mu_m > 0$.   No arbitrage between selling to the government and selling to the market in $t = 1$ means that $p_g + p_2^g = p_m + p_2^m$, so we have
\begin{equation}\label{eq:arg1}
p_2^g = (p_2^m - p_g) + p_m > p_m,
\end{equation}   
where the strict inequality follows from our hypothesis above that  $p_2^m>p_g$.
	Furthermore, since type $\hat{\t}$ must be indifferent between selling in both periods and selling only in $t=2$ at $p_2^m$, we must have 
\begin{equation}\label{eq:arg2}
\hat\t + p_2^m + S = p_m + p_2^m + 2S \implies \hat\t = p_m + S.
\end{equation}
By the zero profit condition, $ p_2^g=\bar\t_g$ and $p_m=\bar\t_m$. Hence, by (\ref{eq:arg1}), we have $\bar\t_g>\bar\t_m$. By (\ref{eq:feas2}), we must have 
\begin{equation}\label{eq:arg3}
\E[\t | \t \le \hat\t] > \bar\t_m= p_m.
\end{equation} 
Then (\ref{eq:arg3}) and (\ref{eq:arg2}) imply 
\begin{equation}\label{eq:arg4}
\E[\t | \t \le \hat\t] +S >  \hat\t.
\end{equation} 
This means that   $\hat \t<\t_0$, by the definition of  $\t_0$.   By (\ref{eq:arg3}), this means that  
$$p_m= \bar\t_m < \E[\t | \t \le \hat\t] <\E[\t | \t \le\t_0]  =p_0,$$
where the last equality follows from  the definition of $p_0$. Suppose a buyer deviates and offers $p' \in (p_m, p_0)$. Since $p' + p_2^m > p_m + p_2^m = p_g + p_2^g$ for any $p' \in (p_m, p_0)$, all types $\t \le \hat\t$ will sell to this deviating buyer. Furthermore, since $p_g + S + \t \le p' + S + p_2^m + S \implies \t \le p' + S$ for any  $p' \in (p_m, p_0)$, types $\t \in (\hat\t, p' + S]$ will sell at the deviation price, too. Since $p' < p_0$, we have $\E[ \t | \t \le p' + S] - p' > 0$, so the deviating buyer will enjoy a strict profit.  We have thus obtained a contradiction to the hypothesis that $p_2^m>p_g$. We also obtain a similar conclusion when $\mu_m=0$ in the main text. We therefore conclude that $p_2^m=p_g$.  Since $\t_2=p_2^m+S$, this in turn implies that  $\t_2 = p_g + S$.   $\square$


 \begin{claim} \label{claim6}
    $\hat\t = \t_0$. 
 \end{claim}

 {\it Proof:}   We know from Claim \ref{claim4} that $\mu_g > 0$.  There are two cases depending on whether   $\mu_m>0$ or $\mu_m=0$.  Consider the former first. 	
 No arbitrage for type $\t \in [0, \hat\t]$ between selling to the government and selling to the market in $t = 1$ implies $p_g + p^g_2 = p_m + p^m_2$, which in turn implies $p^g_2 = p_m$ since $p_g = p^m_2$. By the zero-profit condition, $p^g_2 = \bar\t_g$ and $p_m = \bar\t_m$,  so  by the zero-profit condition $\bar\t_g= \bar\t_m$.  Hence, by \eqref{eq:feas2}, we get  $p_m = \E[ \t | \t \le \hat\t]$. Combining this equality with \eqref{eq:arg2}, we have 
	$$\hat\t =  \E[ \t | \t \le \hat\t] + S.$$
	By definition of $\t_0$, the above equality implies $\hat\t = \t_0$. 
	
Consider next  $\mu_m=0$.  Type $\hat{\t}$ must be indifferent now between selling to the government in $t=1$ followed by selling to the $t=2$ market (with stigma) \emph{and} selling only in $t=2$ market.  Hence, 
 
 \begin{equation}  \label{*}
 p_2^g+p_g+2S=\hat \t +p^2_m+S.
 \end{equation} 
In the proof of Claim \ref{claim5}, we already proved that $p_m^2=p_g$. Hence, (\ref{*}) reduces to
 \begin{equation} \label{**}
p_2^g+ S=\hat \t .
\end{equation} 
 Further,  by the zero profit condition,
  \begin{equation} \label{***}
 p_2^g= \bar\t_m= \E[\t|\t<\hat{\t}],
 \end{equation} 
	where the second equality holds since $\mu_m=0$.  Combining (\ref{**}) and (\ref{***}) gives 
	$$\hat \t= \E[\t|\t<\hat{\t}]+S,$$
which proves that $\hat{\t}=\t_0$. 	  $\square$
 
\section{Proof of Corollary \ref{cor1}} \label{app:cor1}

By Theorem \ref{thm:delay}, the total volume of asset sales is $F(\t_0)+F(p_g+S)$ in any DSE, which equals that under a secret bailout.  It also exceeds the total volume of asset sales in  SSE, $ F(\hat\t)+F(p_g+S)$, since $\hat{\t}<\t_0$.  	  $\square$

\section{Proof of Lemma \ref{lem:wel_fn}} \label{app:lem2}  Suppose the government offers a menu of purchase prices,   $\mathcal P_g\subset [p_0, \overline P]$, for some  arbitrary large real number $\overline P$. Fix any resulting equilibrium whose outcome is  $(Q,T)$.  Let $t_1(\t)$ and $t_2(\t)$, measurable functions of $\t$ with respect to $F(\cdot)$, denote the payment a firm with type $\t$ receives  in $t=1$ and in $t=2$, respectively, in that equilibrium.  Obviously,  $T(\t)=t_1(\t) +t_2(\t)$.   
As presumed by the lemma, assume that the measure of types 
that accept some government bailout offer is positive.

We first prove Part (i). Similar to Lemma \ref{lem:char}, given the single-crossing property, incentive compatility requires that $Q(\t)$ is nondecreasing in $\t$. The ``non-rationing'' assumption then implies that $Q$ must take the form described in (\ref{eq:general-Q}), with two cutoffs $\hat{\t}\le \t_2$. 

 It now remains to prove that $\hat{\t}\le \t_0$, as claimed in Part (ii). 
By definition, the types $\t< \hat \t$ sell in both periods.   Without loss, we consider $\hat \t>0$ (or else we are done).  To begin, define $\overline p:=\sup_{\t<\hat \t} t_1(\t)$, and let $\overline p_2$ be a limit point  of $\{t_2(\t^n)\}$ along a sequence $\{\t^n\}$ such that $t_1(\t^n)\to \overline p$ as $n\to \infty$.\footnote{As will be seen shortly, ``no arbitrage'' will mean that there is a unique limit $\overline p_2$, which will be an infimum of all $t_2(\t)$ among $\t<\t_2$.}  In words, $\overline p_2$ is the price offered in $t=2$ to the firms that chose $\overline p$.
Next, let $p^{**} :=\sup\{t(\t): \exists \t \mbox{ such that } Q(\t)=1\}$ be the highest single sale price; namely, the price any firm receives (in either period) among those who sell only once. This price is finite since all offers made in equilibrium are bounded.  We first conclude:
\begin{align}  \label{eq:p**}
	p^{**}\ge  \overline p.		
\end{align}
This follows since any firm selling only once has an option to sell in $t=1$ via the program that offers the highest price.

By the  ``no arbitrage'' condition, we must have,  for all $\t<\hat\t$,
\begin{align}\label{eq:noarb}
	\overline p+\overline p_2= t_1(\t)+t_2(\t),
\end{align}
which implies 
\begin{align}\label{eq:noarb1}
	\overline p_2\le t_2(\t),
\end{align}
for all $\t<\hat\t$, since $\overline p\ge t_1(\t)$.  Since type $\hat \t$ must weakly prefer selling in both periods to selling only in one period, it follows from (\ref{eq:noarb}) that
\begin{align}\label{eq:thetahat}
	\overline p+	\overline p_2+2S\ge \hat \t +p^{**}+S.
\end{align}

Combining our observations so far, we have
\begin{align}\label{eq:ineq}
	p^{**}+\overline p_2+2S\ge \overline p+\overline p_2+2S\ge \hat \t +p^{**}+S
\end{align}
where the first inequality follows from (\ref{eq:p**}) and the second inequality follows from (\ref{eq:thetahat}).  The inequality (\ref{eq:ineq}) yields:
\begin{align}\label{eq:ineq1}
	\overline p_2 +S\ge \hat \t.
\end{align}

We now have 
\begin{align} \label{eq:ineq2}
	\E[\t|\t\le \hat\t]  = \frac{\int_0^{\hat\t}\t f(\t) d\t}{F(\hat \t)} =\frac{\int_0^{\hat\t}t_2(\t) f(\t) d\t}{F(\hat \t)} \ge\frac{\int_0^{\hat\t}\overline p_2 f(\t) d\t}{F(\hat \t)}= \overline p_2,
\end{align}
where the second equality follows from the zero profit condition that must be satisfied by market offers in $t=2$ and the inequality follows from  (\ref{eq:noarb1}).   It follows from 
(\ref{eq:ineq1}) and  (\ref{eq:ineq2}) that
$$\E[\t|\t\le \hat \t]+S\ge \hat \t,$$
which in turn implies that 
$$\hat \t\le \t_0.$$ 

We now prove Part (ii).  Interpreting the outcome $(Q, T)$ as a direct mechanism, let
$$u(\tilde\t | \t):= T(\tilde\t) + \t (2 - Q(\tilde\t)) + SQ(\tilde\t)$$ denote type-$\t$ firm's payoff when it reports its type as $\tilde\t$ (or acts as if its type is $\tilde \t$ in the equilibrium). Since the mechanism must be incentive compatible for all types, we have $$u(\t | \t) \ge u(\tilde\t | \t),$$ for all $\tilde\t, \t \in [0, 1]$. Let $u(\t) := u(\t | \t)$. Since the participation constraint must be satisfied for all types (in any equilibrium), we also have $u(\t) \ge 2\t$ for all $\t \in [0, 1]$. Let $\overline{u} = u(1)$ be the highest-type ($\t = 1$) firm's payoff in equilibrium. Then, as is standard in mechanism design analysis, one can apply the envelope theorem to obtain 
$$u(\t) = \overline{u} - \int^1_{\t} (2 - Q(s)) ds.$$

The welfare is the sum of the payoffs for firms and buyers minus the deadweight loss from a deficit run by the government. Consequently, the welfare is  written as:
$$\int^1_0 \left[ u(\t) + (\t Q(\t) - T(\t)) \right] f(\t) d\t - \l  \int_{\t \in \T_g} (t_1(\t) - \t) f(\t) d\t.$$ Since buyers in the market must break even in each period, the budget deficit incurred by the government equals the total budget deficit.  That is,
 $$\int_{\t \in \T_g} (t_1(\t) - \t)f(\t) d\t = \int^1_0 (T(\t) - \t Q(\t)) f(\t) d\t.$$
Substituting this expression, the welfare is $$\int^1_0 \left[ u(\t) + (\t Q(\t) - T(\t)) - \l (T(\t) - \t Q(\t)) \right] f(\t) d\t.$$ Substituting 
 $u(\t) = \overline{u} - \int^1_{\t} (2 - Q(s)) ds$ and $u(\t)=T(\t) + \t (2 - Q(\t)) + SQ(\t)$ into the welfare and integrating by parts, we obtain (\ref{eq:wel}).

\section{Proof of Theorem \ref{thm:wel_com}}  \label{app:thm5} Before proceeding, it is useful to establish two preliminary lemmas. 

\begin{lemma} \label{lem:pre1} Fix any equilibrium outcome under a bailout policy,  $(Q,T)$, with cutoffs $(\hat{\t}, \t_2)$ that leaves zero rents for the highest type firm. There exists a (possibly infeasible) outcome $(\tilde Q, \tilde T)$, with $\tilde Q(\t)=2 $ for $\t\le \tilde \t$ and $\tilde Q(\t)=0$ for $\t>\tilde\t$, that entails the same net budget deficit (or surplus) as $(Q,T)$ but entails weakly more trade---that is, $2F(\tilde{\t})\ge F(\hat \t)+F(\t_2)$.
\end{lemma}

\begin{proof} 
Given the continuity of $F$, it suffices to prove that an outcome $(\tilde Q', \tilde{T}')$, with $\tilde Q'(\t)=2 $ for $\t\le \tilde \t'$ and $\tilde Q'(\t)=0$ for $\t>\tilde\t$, where $2F(\tilde{\t})=F(\hat \t)+F(\t_2)$, entails weakly lower budget deficit than $(Q,T)$. This follows from computation of the budget deficit under $(Q,T)$:
	\begin{align*}
     &\int^1_0 \left\{\t Q(\t) - T(\t)\right\} f(\t) d\t	\\
   = &\int^1_0  \left( \frac{F(\t)}{f(\t)}-S\right)Q(\t)f(\t) d\t	\\		
=  &  \int^1_0  \left( \frac{F(\t)}{f(\t)}-S\right)\tilde Q'(\t)f(\t) d\t+  \int^1_0  \left( \frac{F(\t)}{f(\t)}-S\right)(Q(\t)-\tilde Q'(\t))f(\t) d\t	\\
=  &  \int^1_0  \left( \frac{F(\t)}{f(\t)}-S\right)\tilde Q'(\t)f(\t) d\t-
  S   \int^1_0      (Q(\t)-\tilde Q' (\t)) f(\t) d\t  + \int^1_0    \frac{F(\t)}{f(\t)} (Q(\t)-\tilde Q' (\t))  f(\t) d\t	\\
  =   &  \int^1_0  \left( \frac{F(\t)}{f(\t)}-S\right)\tilde Q'(\t)f(\t) d\t +   \int^1_0    \frac{F(\t)}{f(\t)} (Q(\t)-\tilde Q' (\t))  f(\t) d\t	\\
 \ge  &  \int^1_0  \left( \frac{F(\t)}{f(\t)}-S\right)\tilde Q'(\t)f(\t) d\t +   \int^1_0 \frac{F(\tilde\t)}{f(\tilde\t)} (Q(\t)-\tilde Q' (\t))  f(\t) d\t	\\
= &  \int^1_0  \left( \frac{F(\t)}{f(\t)}-S\right)\tilde Q'(\t)f(\t) d\t +   \int^1_0 \frac{F(\tilde\t)}{f(\tilde\t)} (Q(\t)-\tilde Q' (\t))  f(\t) d\t	\\
= &  \int^1_0  \left( \frac{F(\t)}{f(\t)}-S\right)\tilde Q'(\t)f(\t) d\t  	\\
=   &\int^1_0 \left\{\t \tilde Q'(\t) - \tilde T'(\t)\right\} f(\t) d\t.
	\end{align*}
The first and  last equalities follow from the standard envelope argument (e.g., employed in the last part of the proof of Lemma \ref{lem:wel_fn}-(ii)), plus the fact that $\overline u=2$ (i.e., the highest type enjoys no rent under both outcomes).  The second equality follows from adding and subtracting the same integral term.  The third and the second last equalities follow from  $\int^1_0 (\tilde Q'(\t) - Q(\t)) f(\t) d\t = 0$, an implication of the assumption $F(\hat{\t})+F(\t_2)=2F(\tilde{\t})$.  The inequality follows 
	from the fact that $\frac{F(\t)}{f(\t)} $ is strictly increasing in $\t$ (an implication of the log concavity assumed throughout) and that $Q(\t)-\tilde Q' (\t)\ge   0$ for   $\t\ge\tilde{\t}'$ and $  Q(\t)-\tilde Q' (\t) \le  0$ for  $\t\le \tilde{\t}'$.  We have shown that the budget deficit is reduced (at least weakly) under $(\tilde Q',\tilde T')$ while the overall trade volume remains the same. \end{proof}

\begin{lemma} \label{lem:pre} Suppose an equilibrium outcome under a bailout policy,  $(Q,T)$, entails a strictly budget surplus.  Then, it is stictly welfare-dominated by the {\it laissez-faire} outcome under no bailout policy.
\end{lemma}
\begin{proof} By Lemma \ref{lem:pre1}, we know that there is alternative outcome $(\tilde Q,\tilde T)$ which entails the same amount of budget surplus as $(Q,T)$ and entails (weakly) higher welfare. It suffices to show that the latter outcome $(\tilde Q,\tilde T)$ is welfare inferior to the {\it laissez-faire} outcome with common cutoff $\t_0$ across both periods.  To this end, since budget surplus does not contribute to welfare directly, it  suffices to show that the trade volume under  $(\tilde Q,\tilde T)$ is strictly lower than that under  {\it laissez-faire}, or more specifically $\tilde \t<\t_0$.  This follows from the fact that the budget surplus under $(\tilde Q,\tilde T)$ is strictly positive:
\begin{align*}
0< &\int^1_0 \left\{\tilde T(\t)-\t \tilde Q(\t) \right\} f(\t) d\t	\\
= &\int^1_0  \left(S- \frac{F(\t)}{f(\t)}\right)\tilde Q(\t)f(\t) d\t	\\		
= &\int^1_0 S\tilde  Q(\t)f(\t)d\t -  \int^1_0  F(\t)\tilde  Q(\t)d\t 	\\
=  &2F(\tilde{\t})S -2\tilde{\t} F(\tilde{\t}) +2\int^{\tilde{\t}}_0 \t f(\t) d\t 	\\
= &2F(\tilde{\t}) \left\{ S+\E\left[\t|\t\le \tilde{\t}\right] - \tilde{\t}  \right\},   
\end{align*}	
where the first equality follows from  the substitution for $\tilde T$ and the standard envelope argument that has been invoked repeatedly earlier, and the third equality follows from integration by parts.  Given the definition of $\t_0$, the above inequality implies $\tilde{\t}<\t_0$,   Since $2F(\tilde{\t})<2F(\t_0)$, the {\it laissez-faire} under no bailout welfare-dominates $(\tilde Q,\tilde T)$, and hence $(Q,T)$. \end{proof}

We are now in a position to prove the theorem.

\paragraph{\bf Proof of Part (i):}  There are two cases.  Consider first $\lambda>\hat{\l}$. In this case, the claimed optimal outcome $(Q^*,T^*)$ coincides with the {\it laissez-faire} under no bailouts.   In light of Lemma \ref{lem:pre}, it suffices to show that  $(Q^*,T^*)$ welfare dominates any equilibrium outcome $(Q,T)$ that entails nonnegative budget deficit.  By Lemma \ref{lem:pre1}, there exists $(\tilde Q,\tilde T)$ that entails the same budget deficit as  $(Q,T)$, but yields a weakly higher welfare.  Since $(\tilde Q,\tilde T)$ entails a nonnegative budget deficit, by the argument in the proof of Lemma \ref{lem:pre}, we must have $\tilde \t\ge\t_0$.  Use (\ref{eq:wel}) to write the welfare difference between the {\it laissez-faire} and $(\tilde Q,\tilde T)$:
 	\begin{align*}
 &W(Q^*)-W(\tilde Q)\\
 = &\int_0^1  \left((1+\l)S-\l\frac{F(\t)}{f(\t)}\right) [Q^*(\t)-\tilde Q(\t)] f(\t) d\t\\
  	=& - 2 \int_{\t_0}^{\tilde{\t}} \left((1+\l)S-\l\frac{F(\t)}{f(\t)}\right) f(\t) d\t\\
  	\ge & 0,
 \end{align*}
where the inequality follows from the fact that $ (1+\l)S-\l\frac{F(\t)}{f(\t)}<0$ for $\t> \t_0$.  The inequality is strict if $(Q,T)$ does not coincide with the {\it laissez-faire} outcome.  

Consider next  $\lambda\le\hat{\l}$.  We first prove that $(Q^*,T^*)$ welfare-dominates any equilibrium outcome $(Q,T)$, say with cutoffs $\hat{\t}\le \t_2$, that entails nonnegative budget deficit.  This is proven by computing the difference in welfare a la (\ref{eq:wel}):
 	\begin{align*}
&W(Q^*)-W(\tilde Q)\\
= &\int_0^1  \left((1+\l)S-\l\frac{F(\t)}{f(\t)}\right) [Q^*(\t)-\tilde Q(\t)] f(\t) d\t\\
=& \int_0^{\t_0} \left((1+\l)S-\l\frac{F(\t)}{f(\t)}\right)(2-\tilde Q(\t)) f(\t) d\t+     
\int_{\t_0}^{\t_2^*} \left((1+\l)S-\l\frac{F(\t)}{f(\t)}\right)(1-\tilde Q(\t)) f(\t) d\t
\\&\quad
- \int_{\t_2^*}^1 \left((1+\l)S-\l\frac{F(\t)}{f(\t)}\right)\tilde Q(\t) f(\t) d\t  \\
\ge& 0.
\end{align*}
The inequality follows since $(1+\l)S-\l\frac{F(\t)}{f(\t)}\ge 0$ if and only if $\t\le \t_2^*$, since $Q(\t)\le 2$ for all $\t$,  and since  $Q(\t)\le 1$ for all $\t\ge \t_0$, per Lemma \ref{lem:wel_fn}-(i).  
The proof is complete by noting that  the {\it laissez-faire} outcome entails nonnegative deficit and has been shown to be inferior to $(Q^*,T^*)$ and that, by Lemma \ref{lem:pre}, the {\it laissez-faire} outcome in turn welfare-dominates any equilibrium outcome  $(Q,T)$  that entails strict budget surplus.  Note that if $Q$ differs from $Q^*$ for a positive measure of $\t$'s, then the dominance is strict.

\paragraph{\bf Proof of Parts (ii) and (iii):}  Part (ii) follows directly from Theorem \ref{thm:secrecy} and Theorem \ref{thm:delay} (or equivalently, Corollary \ref{cor1}).  Part (iii) follows from Theorem \ref{thm:nodelay}, which shows that the  trade rule $Q$ of the SSE 
has  $\hat\t<\t_0$, so it differs from $Q^*$ for a positive measure of types, and from the fact that $Q^*$ is uniquely optimal.  

\paragraph{\bf Proof of Part  (iv):} 

	Suppose there is a menu of bailout offers, $\mathcal{P}_g\subset [p_0, \overline P]$, for some  arbitrarily large number $\overline P$. Fix any resulting equilibrium with an outcome $(Q, T)$ characterized by cutoffs $0 < \hat\t \le \t_2 \le 1$ such that $Q(\t) = 2$ for all $\t \in [0, \hat\t]$, $Q(\t) = 1$ for all $\t \in (\hat\t, \t_2]$, and $Q(\t) = 0$ for all $\t > \t_2$. Recall that $T(\t)$ can be rewritten as $T(\t) = t_1(\t) + t_2(\t)$, where $t_1(\t)$ and $t_2(\t)$, measurable functions with respect to $F(\t)$, are payments a firm with type $\t$ receives from the asset sale in $t = 1$ and $t = 2$, respectively. Note that
	\begin{align}	\label{eq:no_arb_pf_thm5d}
	t_1(\t) + t_2(\t) = t_1(\t') + t_2(\t') \;\, \mbox{for all}\;\, \t, \t' \in [0, \hat\t]:
	\end{align}
	if not, one type that gets a strictly lower payment must mimic the other type $\t' \le \hat\t$ for a strictly higher payoff.
	
	Next, recall the notations $\overline{p}$ and $\overline{p}_2$ from the proof of Lemma \ref{lem:wel_fn}-(a), where $\overline{p} := \sup_{\t \le \hat \t} t_1(\t)$, and $\overline p_2$ is a limit point of $\{t_2(\t^n)\}_{n \in \mathbb{N}}$ along a sequence $\{\t^n\}_{n \in \mathbb{N}} \subset [0, \hat\t]$ such that $t_1(\t^n)\to \overline p$ as $n \to \infty$. Since $\t^n \in [0, \hat\t]$ for every $n \in \mathbb{N}$, it follows from \eqref{eq:no_arb_pf_thm5d} that $t_1(\t^n) + t_2(\t^n) = t_1(\t) + t_2(\t)$ for every $\t \in [0, \hat\t]$ and $n \in \mathbb{N}$, which implies
	\begin{align}	\label{eq:no_arb_lim_pf_thm5d}
	\overline{p} + \overline{p}_2  = \lim_{n \rightarrow \infty} (t_1(\t^n) + t_2(\t^n)) = t_1(\t) + t_2(\t)  \;\, \mbox{for all}\;\, \t \in [0, \hat\t].
	\end{align}
	Moreover, define $p^{**}$ as the price offer the firms with types $\t \in (\hat\t, \t_2]$ accept in either $t = 1$ or $t = 2$. Note that $T(\t) = T(\t')$ for all $\t, \t' \in (\hat\t, \t_2]$; or else, one type strictly prefers mimicking the other type that receives a strictly higher payment. 
	
	In addition,  let $\T_m$ denote the set of firms that sell the asset to the market in $t = 1$. If $\T_m \neq \emptyset$, we have $t_1(\t) = p_m := \E[ \t | \t \in \T_m]$ for any $\t \in \T_m$ since private buyers get zero expected profit in equilibrium. Note that firms that do not accept any bailout offer in $t = 1$ are pooled together by buyers in $t = 2$, and thus receive the same price offer, denoted by $p^m_2$. 
	
	Lastly, let $\T_g$ denote the set of types that accept the bailout offers in $t = 1$ and  {let  $\mathcal{P}^A_g$ denote the support of $\set{t_1(\t) : \t \in \T_g}$.\footnote{{Formally, $\mathcal{P}^A_g:=\{p_g\in \mathcal{P}_g:\Pr\{\t\in\T_g:  t_1(\t)\in (p_g-\epsilon, p_g+\epsilon)\}>0, \forall \epsilon>0\}$. More intuitively, for each  $p_g\in \mathcal{P}^A_g$, there must be a positive measure of types in $\T_g$ that choose $p_g$ or bailout offers in the arbitrarily close neighborhood of $p_g$.}} By the hypothesis of the statement, {\it (i)} the measure of $\T_g$ is positive; and {\it (ii)} $\mathcal{P}^A_g$ is not a singleton.} The proof proceeds in several steps.		\\
	
	
\noindent \emph{Step 1.} $\T_m \subset [0, \hat\t]$.	

{Suppose to the contrary that there is a positive measure of firm types $\t \in (\hat\t, \t_2] \cap \T_m$ that sell at $p_m$ in $t = 1$ but refuses to sell in $t = 2$. The fact that $\t\in \T_m$ means that 
	\begin{align} \label{eq:step1}
	p_m + S \ge \t,
	\end{align}
	since type $\t$ must at least weakly prefer selling to the market in $t=1$ (from which it will   earn $p_m + S + \max\set{\t, p^m_2 + S}$) to holding out its asset in $t=1$ (which will yield $\t + \max\set{\t, p^m_2 + S}$).\footnote{Recall that all firms that do not accept any bailout offer in $t = 1$ receive the same price offer $p^m_2$ in $t = 1$.}  Meanwhile, the fact that $\t \in (\hat\t, \t_2]$  implies that
	\begin{align} \label{eq:step1'}
 \t > p^m_2 + S,
	\end{align}
	since the type $\t \in (\hat\t, \t_2]$ must strictly prefer not to sell at price $p^m_2$ in $t = 2$. Combining  (\ref{eq:step1}) and  (\ref{eq:step1'}), we conclude that
	\begin{align} \label{eq:step1''}
    p_m > p^m_2.
	\end{align}
	Suppose now a buyer in $t = 2$  deviates and offers $p_m$ to the firms that do not sell to the government in $t = 1$. Recall from   (\ref{eq:step1}) that
    $\t \le p_m + S$ for all $\t \in \T_m$. Since $p_m > p^m_2$ by (\ref{eq:step1''}), all types $\t \in \T_m$ would rather sell at the deviation price $p_m$ than $p^m_2$ in $t = 2$. Moreover, $p_m$ breaks even for the deviating buyer since $\E[ \t | \t \in \T_m] = p_m$---the break-even condition for $t=1$ buyers in the initial hypothesized equilibrium. Since $p_m > p^m_2$, the deviating offer $p_m$ will be strictly preferable to a positive measure of firms (relative to their putative equilibrium payoffs) and at least breaks even for the deviating firm. This violates our equilibrium condiction, so we have a contradiction. } 
    $\square$\\

    \noindent \emph{Step 2.}  There exist no bailout offers that are accepted \emph{only} by a positive measure of types in $\T_g \cap (\hat\t,1]$.  Formally, $\tilde{\mathcal{P}}_g=\emptyset$, where $\tilde{\mathcal{P}}_g$ is the support of $\{t_1(\t)\in \mathcal{P}_g:  \T_g \cap (\hat\t,1]\}$.

	Suppose to the contrary that  $\tilde{\mathcal{P}}_g$ is nonempty.    Then, all buyers in $t = 2$ believe that firms that accept  bailout offers $\tilde{\mathcal{P}}_g$ must have types $\t > \hat\t$. Let $\t' \in (\hat\t, \t_2]$ denote the infimum of all types that accept $\tilde{\mathcal{P}}_g$. Suppose a $t = 2$ buyer deviates and offers $p' = \t' - \e$ for a sufficiently small $\e > 0$ to these firms. Then a positive measure of types $\t \in (\t', p' + S]$ that accept $p^j_g \in \tilde{\mathcal{P}}_g$ sell at the deviation price $p'$ in $t = 2$. This implies the average value of these types that sell at $p'$ is bounded below by $\t'$. Since $\t' - p' = \e > 0$, the deviating buyer makes a profit, a contradiction. $\square$	\\

\noindent \emph{Step 3.} $\overline{p}_g = \overline{p}$, where $\overline{p}_g := \sup_{ \t \in \T_g} t_1(\t)$ (and recall $\bar{p} = \sup_{\t \le \hat\t} t_1(\t)$).

	Since every bailout offer accepted in $t = 1$ is accepted by a type-$\t$ firm for some $\t \le \hat\t$ from Step 2, we have $\overline{p}_g \le \overline{p}$. We have $\overline{p}_g = \overline{p}$ trivially, if $\T_m = \emptyset$. Hence, assume $\T_m \neq \emptyset$. 
	To prove $\overline{p}_g = \overline{p}$, suppose to the contrary that $\overline{p}_g < \overline{p}$. Since $p_g \le \overline{p}_g < \overline{p}$ for every $p_g\in \mathcal{P}^A_g$, we must have $p_m = \overline{p}$, which implies $p_m > p_g \ge p_0$. Since $\T_m \subset [0, \hat\t]$ from Step 1, we have $t_1(\t) + t_2(\t) = p_m + p^m_2$ for every $\t \in \T_m$. Since firms with types $\t \le \hat\t$ receive the same payment over two periods by \eqref{eq:no_arb_pf_thm5d}, we must have
	\begin{align}	\label{eq:indiff_m_pf_thm5d}
	p_m + p^m_2 = t_1(\t) + t_2(\t) \;\, \mbox{for all} \;\, \t \le \hat\t.
	\end{align}
	In $t = 2$, all types $\t \in \T_m$ and possibly some types $\t' \in (\hat\t, \t_2]$ sell to the market in $t = 2$ at $p^m_2$. Since $\T_m \subset [0, \hat\t]$ and private buyers get zero expected profit, we must have $p^m_2 \ge \E[ \t | \t \in \T_m] = p_m$.
	
	Moreover, we must have $t_1(\t) \le p^{**}$ for every $\t \le \hat\t$, where $p^{**} = T(\t)$ for all $\t \in (\hat\t, \t_2]$: otherwise, every firm with $\t \in (\hat\t, \t_2]$ would rather sell at $t_1(\t') > p^{**}$ for some $\t' \le \hat\t$ and not sell in $t = 2$. This implies $p_g < p_m = \overline{p} \le p^{**}$ for all $p_g\in \mathcal{P}^A_g$. It then follows that no firm with $\t > \hat\t$ would accept any bailout offer in $t = 1$, which implies $\T_g \subset [0, \hat\t]$.  Since $t_1(\t) \le \overline{p}_g < p_m$ for every $\t \in \T_g$ and $\T_g \subset [0, \hat\t]$, it follows from \eqref{eq:indiff_m_pf_thm5d} that 
	$$t_2(\t) > p^m_2 \ge p_m \;\,\mbox{for all}\;\, \t \in \T_g.$$ 	
	Furthermore, since $\T_m, \T_g \subset [0, \hat\t]$ and every firm with type $\t \le \hat\t$ must sell either to the government or to the market in $t = 1$, we must have 
	$$[0, \hat\t] = \T_g \cup \T_m.$$ 
		
	Since any price that private buyers' offer to the firms must break even, we have $\E[ \t  | \t \in \T_g] = \E[t_2(\t) | \t \in \T_g] > p_m,$ where the strict inequality follows from the observation that $t_2(\t) > p_m$ for every $\t \in \T_g$. Lastly, since $p_m > p_g \ge p_0$ for any bailout offer $p_g\in \mathcal{P}^A_g$, we have $\E[ \t | \t \in \T_m] = p_m > p_0$. Since $[0, \hat\t] = \T_g \cup \T_m$, we have
	$$\E[ \t | \t \le \hat\t] = \E[ \t | \t \in \T_g \cup \T_m] \ge p_m > p_0,$$
	which implies $\hat\t > \t_0$. This is a contradiction since $\hat\t \le \t_0$ by Lemma \ref{lem:wel_fn}. $\square$ \\ 
	
	

\noindent \emph{Step 4.} $\hat\t < \t_0$.

	We prove this for the case $\T_m \neq \emptyset$ first. From Step 1, we have $\T_m \subset [0, \hat\t]$, which implies $t_1(\t) + t_2(\t) = p_m + p^m_2$ for every $\t \le \hat\t$ from \eqref{eq:no_arb_pf_thm5d}. Furthermore, it follows from the definition of $p^{**}$ that $p^m_2 \le p^{**}$: otherwise, any type-$\t$ firm with $\t \in (\hat\t, \t_2]$ strictly prefers selling only in $t = 2$ at price $p^m_2$. Since type-$\hat\t$ firm weakly prefers selling in both periods to selling only in one period at price $p^{**}$, we must have
	$$t_1(\hat\t) + t_2(\hat\t) + 2S = p_m + p^m_2 + 2S \ge \hat\t + p^{**} + S \ge \hat\t + p^m_2 + S,$$
	which implies 
	\begin{align}	\label{eq:pf_thm5d_1}
	p_m + S \ge \hat\t.
	\end{align}
	
	Moreover, we must have $\Pr\{\t \in \T_g \cap [0, \hat\t]\} > 0$. If $\Pr\{\t \in \T_g \cap [0, \hat\t]\} = 0$, firms with types $\t \in \T_g$ almost surely accept the bailout but do not sell in $t = 2$. This implies that a single bailout offer $p_g = p^{**}$ must be accepted with probability 1, which contradicts the hypothesis (ii) that $\mathcal{P}^A_g$ is not a singleton.
	
	Furthermore, since $t_1(\t) \le \overline{p}$ by definition of $\overline{p}$, by \eqref{eq:no_arb_lim_pf_thm5d}, we have
	\begin{align}	\label{eq:t_2_bd_pf_thm5d}
	t_2(\t) \ge \overline{p}_2 \;\, \mbox{for every} \;\, \t \le \hat\t.
	\end{align} 
	Next, recall $\overline{p}_g = \overline{p}$ from Step 3, which implies $\overline{p}_g + \overline{p}_2 = t_1(\t) + t_2(\t)$ for all $\t \in \T_g \cap [0, \hat\t]$. From Step 2 and the assumption that $\mathcal{P}^A_g$ is not a singleton, we know that, for any $p_g \in \mathcal{P}^A_g$, there exists a positive measure of $\t \in \T_g \cap [0, \hat\t]$ such that $t_1(\t) \neq p_g < \overline{p}_g$. Combining this observation with $\Pr\{\t \in \T_g \cap [0, \hat\t]\} > 0$, we have $t_1(\t) < \overline{p}_g$, and therefore, $t_2(\t) > \overline{p}_2$ for a positive measure of $\t \in \T_g \cap [0, \hat\t]$. 
	 
	 In addition, we know $\overline{p}_g \le p^{**}$; otherwise, every firm with $\t \in (\hat\t, \t_2]$ will deviate and accept some bailout offer $p_g > p^{**}$ in $t = 1$. Since type-$\hat\t$ firm weakly prefers selling two units of the assets to selling one unit, we have
	$$t_1(\hat\t) + t_2(\hat\t) + 2S = \overline{p}_g + \overline{p}_2 + 2S \ge \hat\t + p^{**} + S \ge \hat\t + \overline{p}_g + S,$$
	which implies
	\begin{align}	\label{eq:p2_high_pf_thm5d}
	\overline{p}_2 + S \ge \hat\t.	
	\end{align}
	Combining the observations \eqref{eq:pf_thm5d_1} and \eqref{eq:p2_high_pf_thm5d} together, we have
	\begin{align}
	(p_m \wedge \overline{p}_2) + S \ge \hat\t.	\label{eq:pf_thm5d_2}
	\end{align}
		
	Lastly, all types $\t \in [0, \hat\t] \setminus \T_m$ must accept some bailout offers in $t = 1$ and sell to the market in $t = 2$, which implies $[0, \hat\t] \setminus \T_m = \T_g \cap [0, \hat\t]$. Since every private buyer that purchases the asset must get zero expected payoff, we have
	\begin{equation}	\label{eq:pf_thm5d_3}
	\begin{aligned}	
	\E[ \t | \t \le \hat\t] &= \frac{ \int_{\t \in [0, \hat\t] \setminus \T_m} \t dF + \int_{\t \in \T_m} \t dF }{F(\hat\t)} =  \frac{ \int_{\t \in \T_g \cap [0, \hat\t]} t_2(\t) dF + \int_{\t \in \T_m} t_1(\t) dF }{F(\hat\t)} \\
				   &> \frac{ \int_{\t \in \T_g \cap [0, \hat\t]} \overline{p}_2 dF + \int_{\t \in \T_m} p_m dF }{F(\hat\t)} \ge  p_m \wedge \overline{p}_2.
	\end{aligned}
	\end{equation}
	The strict inequality follows from the observations {\it (a)} $t_2(\t) \ge \overline{p}_2$ for every $\t \in \T_g \cap [0, \hat\t]$ as in \eqref{eq:t_2_bd_pf_thm5d}; {\it (b)} $t_2(\t) > \overline{p}_2$ for a positive measure of $\t \in \T_g \cap [0, \hat\t]$; {\it (c)} $t_1(\t) = p_m$ for every $\t \in \T_m$; and {\it (d)} $Pr(\t \in \T_g \cap [0, \hat\t]) > 0$. From \eqref{eq:pf_thm5d_2} and \eqref{eq:pf_thm5d_3}, we have
	$$\E[ \t | \t \le \hat\t] + S > (p_m \wedge \overline{p}_2) + S \ge \hat\t,$$
	which implies $\hat\t < \t_0$. 
	
	Next, consider the case $\T_m = \emptyset$. Since there is no asset trading at the private market in $t = 1$, every type-$\t$ firm with $\t \in [0, \hat\t]$ must sell to the government in $t = 1$. By applying the same logic previously used in the proof for the case $\T_m \neq \emptyset$, one can find that {\it (a)} $t_2(\t) \ge \overline{p}_2$ for every $\t \le \hat\t$, where the inequality strictly holds for a positive measure of $\t \le \hat\t$; and {\it (b)} $\overline{p}_2 + S \ge \hat\t$. Furthermore, it follows from \eqref{eq:no_arb_lim_pf_thm5d} and Step 3 that $\overline{p} + \overline{p}_2 = \overline{p}_g + \overline{p}_2 = t_1(\t) + t_2(\t)$ for every $\t \le \hat\t$. Combining all observations so far, we have 
	$$\E[ \t | \t \le \hat\t] = \E[ t_2(\t) | \t \le \hat\t] > \overline{p}_2,$$
	where the equality follows from the fact that every private buyer gets zero expected payoff. Hence, we have $\E[ \t | \t \le \hat\t] + S > \overline{p}_2 + S \ge \hat\t$, which again implies $\hat\t < \t_0$. $\square$	\\
	
\noindent \emph{Step 5.}  The hypothesized equilibrium is strictly suboptimal.

  Since $\hat\t < \t_0$ from Step 4, this outcome must yield a strictly lower welfare than the optimal trading rule $(\hat\t, \t_2) = (\t_0, \max\set{\t_0, \t^*_2})$ by Theorem \ref{thm:wel_com}-(i). $\square$

\end{document}